\documentclass[11pt]{article}
\usepackage{hyperref}
\hypersetup{
    colorlinks = true,
    linkcolor = {red}
}
\usepackage{pset}

\title{Linear Programming Bounds for Randomly Sampling Colorings}
\author{Sitan Chen\thanks{This work was supported in part by NSF CAREER Award CCF-1453261 and NSF Large CCF-1565235.} \\MIT  \and Ankur Moitra\thanks{This work was supported in part by NSF CAREER Award CCF-1453261, NSF Large CCF-1565235, a David and Lucile Packard Fellowship, and an Alfred P. Sloan Fellowship.}\\ MIT}

\newcolumntype{Y}{>{\centering\arraybackslash}X}

\renewcommand{\E}{\mathbb{E}}
\newcommand{\Bad}[1]{\textsc{Bad}$_{#1}$}
\newcommand{\Sing}[1]{\textsc{Sing}$_{#1}$}
\newcommand{\Good}[1]{\textsc{Good}$_{#1}$}
\newcommand{\Badend}[1]{\textsc{BadEnd}$_{#1}$}
\newcommand{\Goodend}[1]{\textsc{GoodEnd}$_{#1}$}
\newcommand{\Tstop}{T_{\text{stop}}}
\renewcommand{\P}{\mathbb{P}}

\begin{document}

\maketitle

\begin{abstract}
\normalsize
Here we study the problem of sampling random proper colorings of a bounded degree graph. Let $k$ be the number of colors and let $d$ be the maximum degree. In $1999$, Vigoda \cite{vigoda2000improved} showed that the Glauber dynamics is rapidly mixing for any $k > \frac{11}{6} d$. It turns out that there is a natural barrier at $\frac{11}{6}$, below which there is no one-step coupling that is contractive, even for the flip dynamics.  

We use linear programming and duality arguments to guide our construction of a better coupling. We fully characterize the obstructions to going beyond $\frac{11}{6}$. These examples turn out to be quite brittle, and even starting from one, they are likely to break apart before the flip dynamics changes the distance between two neighboring colorings. We use this intuition to design a variable length coupling that shows that the Glauber dynamics is rapidly mixing for any $k\ge (\frac{11}{6} - \epsilon_0)d$ where $\epsilon_0 \geq 9.4 \cdot 10^{-5}$. This is the first improvement to Vigoda's analysis that holds for general graphs. 
\end{abstract}

\thispagestyle{empty}
\setcounter{page}{0}

\newpage

\section{Introduction}

\subsection{Background}

Here we study the problem of sampling random proper colorings of a bounded degree graph. More precisely, let $k$ be the number of colors and let $d$ be the maximum degree. A long-standing open question is to give an algorithm that works for any $k \geq d + 2$, when the space of proper colorings is first connected. Despite a long line of intensive investigation \cite{jerrum1995very, salas1997absence, dyer2003randomly, dyer2004randomly, hayes2003randomly, hayes2003non, molloy2004glauber, hayes2005coupling, frieze2006randomly, frieze2007survey}, the best known bounds are quite far from the conjecture. 

In fact, there is a natural Markov chain called the {\em Glauber dynamics} that is widely believed to work: in each step, choose a random node and recolor it with a random color not appearing among its neighbors. It is easy to see that its steady state distribution is uniform on all proper $k$ colorings, again provided that $k \geq d + 2$. It is even conjectured that on an $n$ node graph, the mixing time is $O(n \log n)$ which would be tight \cite{hayes2005general}. We remark that rapidly mixing Markov chains for sampling random colorings immediately give a fully polynomial randomized approximation scheme (FPRAS) for counting the number of proper colorings. There is also interest in this question in combinatorics \cite{brightwell2002random} and in statistical physics, where it corresponds to approximating the partition function of the zero temperature anti-ferromagnetic Potts model \cite{potts1952some}.

Jerrum \cite{jerrum1995very} gave the first significant results and showed that when $k > 2d$ the Glauber dynamics mixes in time $O(n \log n)$. The modern proof of this result is easier and proceeds through {\em path coupling} \cite{bubley1997path}, whereby it is enough to couple the updates between two colorings $\sigma$ and $\tau$ that differ only at a single node $v$ and show that the expected distance between them is strictly decreasing. Then Jerrum's bound follows by comparing how often the distance between the colorings decreases (when $v$ is selected and after the update has the same color in both) vs. how often it increases (when a neighbor of $v$ is selected and recolored in one but not the other). This result is closely related to work in the statistical physics community by Salas and Sokal \cite{salas1997absence} on the {\em Dobrushin uniqueness condition}. 

In a breakthrough work, Vigoda \cite{vigoda2000improved} gave the first algorithm for sampling random colorings that crossed the natural barrier of $2d$. His approach was through a different Markov chain that in addition to recoloring single nodes also swaps the colors in larger alternating components. His chain was a variant of the Wang-Swendsen-Koteck\'{y} (WSK) algorithm \cite{wang1989antiferromagnetic}. The key insight is that the bottleneck in Jerrum's approach \---- when the neighbors of $v$ all have distinct colors \---- can be circumvented by flipping larger components. More precisely, when a neighbor of $v$ is recolored in one chain in a way that would have increased the distance, one can instead match it with the flip of an alternating component of size two in the other chain that keeps the distance the same. But now one needs to couple the flips of larger alternating components in some manner. Vigoda devised a coupling and a choice of flip probabilities that works for any $k > \frac{11}{6} d$. His Markov chain mixes in time $O(n \log n)$ and one can also connect it to Glauber dynamics and prove an $\widetilde{O}(n^2)$ mixing time under the same conditions. This is still the best known bound for general graphs. 

Subsequently, there was a flurry of work on getting better bounds for restricted families of graphs. Dyer and Frieze \cite{dyer2003randomly} considered graphs of logarithmic maximum degree and girth and proved that the Glauber dynamics mixes rapidly whenever $k > \alpha d$ where $\alpha$ is the solution to $\alpha = e^{1/\alpha}$ and numerically $\alpha = 1.763\cdots$. Their approach was to show that under the uniform distribution on proper colorings, the number of colors missing from the neighborhood of $v$ is roughly $k(1-\frac{1}{k})^d$ with high probability. Results like these were later termed {\em local uniformity properties}. They were improved in many directions in terms of reducing the degree and/or the girth requirements \cite{hayes2003randomly, molloy2004glauber, hayes2005coupling, frieze2006randomly, lau2006randomly}, culminating in two incomparable results. Dyer et al. \cite{dyer2004randomly} showed that Glauber dynamics mixes rapidly whenever $k > \beta d$ where $\beta$ is the solution to $(1-e^{-1/\beta})^2 + \beta e^{-1/\beta} = 1$ and numerically $\beta = 1.489\cdots$ for girth $g \geq 6$ and the degree $d$ being a sufficiently large constant. Hayes and Vigoda \cite{hayes2003non} showed rapid mixing for any $k > (1+\epsilon)d$ for any $\epsilon > 0$ provided that the girth $g \geq 11$ and the degree is logarithmic, using an intriguing non-Markovian coupling. 

It is important to emphasize that the types of local uniformity properties being exploited by the works above do {\em not} hold for general graphs \---- e.g. ones with triangles. In fact, even the chromatic number is asymptotically different: Johansson \cite{johansson1996asymptotic} proved that the chromatic number of a triangle free graph is at most $\frac{(9 + o(1)) d}{\log d}$ which was later improved by Molloy \cite{molloy2017list} to $\frac{(1 + o(1)) d }{\log d}$, compared to $d + 1$ for general graphs which is tight. 

There have been many other improvements, but all for special graph families. Through an eigenvalue generalization of the Dobrushin condition, Hayes \cite{hayes2006simple} showed that Glauber dynamics mixes rapidly for $k > d + c \sqrt{d}$ on planar graphs and graphs of constant treewidth. Berger et al. \cite{berger2005glauber} showed rapid mixing on graphs of logarithmic cutwidth, which was strengthened by Vardi \cite{vardi2017randomly} to graphs of logarithmic pathwidth. Some recent papers have studied settings such as bipartite or random graphs \cite{dyer2006randomly, mossel2010gibbs, efthymiou2018sampling}, where it is possible to mix with fewer colors than the maximum degree. Hayes et al. \cite{hayes2015randomly} notably improved the abovementioned result of \cite{hayes2006simple} to show that Glauber dynamics in fact mixes rapidly for planar graphs when $k = \Omega(d/\log d)$. These works all leverage structural properties that hold in restricted settings.

\subsection{Our Results}

Our main result is the first improvement on randomly sampling colorings on general bounded degree graphs since the $1999$ paper of Vigoda \cite{vigoda2000improved}. Specifically, we prove:

\begin{thm}\label{thm:main}
The flip dynamics is rapidly mixing with mixing time $O(n \log n)$, for any $k \geq (\frac{11}{6} - \epsilon_0) d$ where $\epsilon_0 \geq 9.4\cdot 10^{-5}$ is an absolute constant that is independent of $d$.
\end{thm} 

Our proof is guided by linear programming and duality arguments. It is not really a computer assisted proof in the sense that we obtain the improvement over $\frac{11}{6}$ by solving a single linear program that works for all degrees, that leverages various structural tools we prove about the space of proper colorings in general graphs. The starting point for our approach is the observation that choosing the best flip probabilities in the WSK algorithm (i.e. the probability of flipping alternating components of each possible size) can be cast as a linear program, when utilizing Vigoda's greedy coupling \cite{vigoda2000improved}. In this manner, Vigoda's analysis provides a feasible solution. Either surprisingly or unsurprisingly, this feasible solution turns out to be optimal in a strong sense: Not only is it an optimizer to the linear programming relaxation\footnote{Several approximations are made along the way in reaching this linear program, such as restricting to flipping components of size at most $6$ and considering all configurations of colors in the layers around $v$ whether they are realizable or not.}, we in fact prove that there is no one-step coupling that is contractive with respect to the Hamming metric for $k < \frac{11}{6} d$ (Lemma~\ref{lem:tight}). In this sense, Vigoda's threshold of $\frac{11}{6}$ is a natural barrier for a class of analyses. 

However in bounding the mixing time of Markov chains, there is always the hope of utilizing non-Markovian couplings that look into the past or future to obtain better bounds. The first improvements to Vigoda's bound for graphs of large degree and girth by Dyer and Frieze \cite{dyer2003randomly} proceeded in this manner via a burn-in method. Hayes and Vigoda \cite{hayes2003non} gave a sophisticated method based on looking into the future to remove obstructions, again under various assumptions. Non-Markovian couplings have also been used to get $O(n \log n)$ mixing time \---- rather than $\widetilde{O}(n^2)$ \---- for Glauber dynamics with $k = (2 - \epsilon)d$ colors \cite{dyer2001extension, hayes2007variable}. 

We leverage our linear programming formulation for finding a one-step coupling to pass to the dual and exactly characterize the family of local neighborhoods around $v$ that cause such analyses to get stuck at $\frac{11}{6}$. There are already some specific examples along these lines mentioned in \cite{vigoda2000improved}. For our purposes, it is crucial that we have a full characterization because our approach is to use non-Markovian couplings to simultaneously defeat all of these examples. The intuition is that if we can, then by complementary slackness we ought to be able to break the $\frac{11}{6}$ barrier. There are some subtleties to making this work, such as showing that the improvement is a universal constant independent of $d$. 

So how do we construct multi-step couplings? The main idea is the family of tight examples is {\em brittle}, and even starting from one, we are reasonably likely to break it apart before updating $v$ (the node of disagreement) or any of its neighbors. The random process is quite complicated, but we are able to coarsen the state space using the appropriate notions of good and bad (and things in between) states. Ultimately, our coupling is not contractive over one step but when you measure the expected change in distance at a random stopping time (the first time we do not use the identity coupling) it {\em is}. Another way to think about our technique is as an amortized analysis that even though there are pairs of configurations where we cannot break the $\frac{11}{6}$ barrier, the fact that we can for many other pairs (that are likely to be reached even from bad starting points) is enough. 

As in Vigoda's paper \cite{vigoda2000improved}, we obtain the following as an implication of our main result: 

\begin{thm}
The Glauber dynamics is rapidly mixing with mixing time $O(n^2 \log n \log k)$, for any $k \geq (\frac{11}{6} - \epsilon_0) d$ where $\epsilon_0 \geq 9.4\cdot 10^{-5}$.\label{thm:compareglauber}
\end{thm}



We believe that the idea of leveraging insights from linear programming and duality may be more generally useful in constructing non-Markovian couplings. Such techniques have been used in approximation algorithms \cite{goemans1998improved, jain2003greedy} under the name {\em factor revealing LPs}. Here they play a somewhat different role in utilizing duality and complementary slackness to identify and characterize the obstacles to a single step coupling in a principled manner. We remark that before Vigoda's work \cite{vigoda2000improved}, Bubley et al. \cite{bubley1998beating} used linear programming to show that Glauber dynamics is rapidly mixing with five colors on graph with maximum degree three. Their approach required solving several hundred linear programs, and was subject to ``combinatorial explosion" as a function of the degree. This also points to a main challenge going forward: to find appropriate linear programming relaxations for constructing families of couplings that avoid such an explosion and to understand what these relaxations do and do not give up. This can be quite subtle, but we have many geometric tools to build our understanding upon. 
\subsection{Organization}

In Section~\ref{sec:prelim} we give the basic definitions that arise in variable-length path coupling, recall Vigoda's Markov chain, and interpret Vigoda's one-step coupling analysis as implicitly solving a linear program. In Section~\ref{sec:tightexamples} we identify a family of worst-case neighborhoods which is tight for Vigoda's approach and define a certain ``$\gamma$-mixed'' variant of Vigoda's linear program which will form the basis for our coupling analysis. In Section~\ref{sec:coupling} we formulate our variable-length coupling and exhibit a reduction from analyzing this coupling to analyzing our modified linear program for a particular $\gamma$. This $\gamma$ depends on the ``typical'' collection of alternating components containing $v$ when the coupling terminates, and in Section~\ref{sec:rand_process}, we analyze this typical collection to estimate $\gamma$. In Appendix~\ref{app:obsproof} we prove Lemma~\ref{obs:slack}.

\section{Preliminaries}
\label{sec:prelim}

In a graph $G = (V,E)$, for vertex $v\in V$ define $N(v)$ to be the set of neighbors of $v$ and $d(v)$ to be the degree of vertex $v$. Given a coloring $\sigma: V\to[k]$, define $A_{\sigma}(v)$ to be the set of colors \emph{available to $v$}, i.e. the set of colors $c'$ for which no neighbor of $v$ is colored $c'$. Given a Markov chain with transition probability matrix $P$ on finite state space $\Omega$ and initial state $\sigma^{(0)}$, denote the distribution of state $\sigma^{(t)}$ at time $t$ by $P^t(\sigma^{(0)},\cdot)$. Denote the stationary distribution of the Markov chain by $\pi$, and define the \emph{mixing time} $\tau_{mix}$ by $$\tau_{mix}\triangleq \max_{\sigma^{(0)}\in\Omega}\min\{t: \tvd(P^{t'}(\sigma^{(0)},\cdot),\pi)\le 1/2e \ \forall t'\ge t\}.$$

\subsection{The Flip Dynamics}
\label{subsec:introflip}

The Markov chain we use is a variant of the Wang-Swendsen-Koteck\'{y} (WSK) algorithm \cite{wang1989antiferromagnetic} studied in \cite{vigoda2000improved}, which we define below. In a coloring $\sigma$ of a graph $G$, for vertex $v$ and color $c$ let $S_{\sigma}(v,c)$ denote the set of vertices $w$ for which there exists an alternating path between $v$ and $w$ using only the colors $c$ and $\sigma(v)$. Under this definition, $S_{\sigma}(v,\sigma(v)) = \emptyset$. The motivation for this definition is that if $\sigma$ is proper, then if one \emph{flips} $S_{\sigma}(v,c)$, i.e. changes the color of all $\sigma(v)$-colored vertices in $S_{\sigma}(v,c)$ to $c$ and that of all $c$-colored vertices in $S_{\sigma}(v,c)$ to $\sigma(v)$, the resulting coloring is still proper.

\begin{defn}
	Let $\{p_{\alpha}\}_{\alpha\in\N}\in[0,1]^{\N}$ be a collection of \emph{flip probabilities}. The \emph{flip dynamics} is a random process generating a sequence of colorings $\sigma^{(0)},\sigma^{(1)}, \sigma^{(2)},\cdots$ of $G$ where $\sigma^{(0)}$ is an arbitrary proper coloring and $\sigma^{(t)}$ is generated from $\sigma^{(t-1)}$ as follows:\begin{enumerate}
		\item Select a random vertex $v^{(t)}$ and a random color $c^{(t)}$.
		\item Let $\alpha = |S_{\sigma}(v^{(t)},c^{(t)})|$ and flip $S_{\sigma}(v^{(t)},c^{(t)})$ with probability $p_{\alpha}/\alpha$.
	\end{enumerate}

	The reason for the $p_{\alpha}/\alpha$ term is that we have a nice equivalent way of formulating the flip dynamics. Let $\mathcal{S}$ denote the family of all alternating components, i.e. all $S\subset V$ for which there exist $v,c$ such that $S = S_{\sigma}(v,c)$. $\sigma^{(t)}$ is generated from $\sigma^{(t-1)}$ as follows:\begin{enumerate}
		\item Pick any component $S_t\in\mathcal{S}$, each with probability $1/nk$.
		\item Flip $S$ with probability $p_{|S|}$.
	\end{enumerate}
\end{defn}

\begin{lem}[\cite{vigoda2000improved}]
	The flip dynamics are an aperiodic, irreversible, symmetric Markov chain, so its stationary distribution is the uniform distribution over proper colorings of $G$.
\end{lem}

The WSK algorithm corresponds to a choice of $p_{\alpha} = 1$ for all $\alpha\in\N$. For the purposes of path coupling, it turns out one only needs to flip alternating components whose size is at most some absolute constant $N_{max}$ (in Vigoda's Markov chain, $N_{max} = 6$), and this ``local'' nature of the flip dynamics will simplify the analysis. 

\subsection{Variable-Length Path Coupling}

Coupling is a common way to bound the mixing time of Markov chains. A $T$-step coupling for a Markov chain with transition matrix $P$ and state space $\Omega$ defines for every initial $(\sigma^{(0)},\tau^{(0)})\in\Omega^2$ a stochastic process $(\sigma^{(t)},\tau^{(t)})$ such that the distribution of $\sigma^{(T)}$ (resp. $\tau^{(T)}$) is the same as $P^T(\sigma^{(0)},\cdot)$. (resp. $P^T(\tau^{(0)},\cdot)$). The \emph{coupling inequality} states that for any starting point $\sigma^{(0)}$ for the Markov chain, $$\tvd(\sigma^{(t)},\pi)\le\max_{\tau^{(0)}}\Pr(\sigma^{(T)}\neq\tau^{(T)}).$$

We will think of $T$-step couplings as random functions $\Omega^2\to\Omega^2$, so we will denote them by $(\sigma^{(0)},\tau^{(0)})\mapsto(\sigma^{(T)},\tau^{(T)})$, or more succinctly, $(\sigma,\tau)\mapsto(\sigma',\tau')$. So if one can devise a coupling $(\sigma,\tau)\mapsto(\sigma',\tau')$ that $(1-\alpha)$-contracts for some $\alpha > 0$ and metric $d$ on $\Omega$, i.e. that satisfies \begin{equation}\E[d(\sigma',\tau')]\le (1 - \alpha)d(\sigma,\tau),\label{eq:contract}\end{equation} then one can show that $\tau_{mix} = O(T\log(D)/\alpha)$, where $D$ is the diameter of $\Omega$ under $d$.

In complicated state spaces like the space of all proper colorings, it is often tricky to construct couplings that give good bounds on mixing time. Path coupling, introduced in \cite{bubley1997path}, is a useful tool for simplifying this process: rather than define $(\sigma,\tau)\mapsto(\sigma',\tau')$ for all $(\sigma,\tau)\in\Omega^2$, it is enough to do so for a small subset of initial pairs in $\Omega^2$. This is called a \emph{partial coupling}.

For the rest of this subsection, we specialize our discussion of coupling to the setting of sampling colorings. For a graph $G$, let $\Omega^*$ denote the space of all proper colorings of $G$, and let $\Omega = [k]^{V}$ denote the space of all colorings of $G$. Fix any Markov chain over $\Omega$ whose stationary distribution is the uniform distribution over $\Omega^*$, e.g. the Glauber or flip dynamics, and denote it by $\sigma\mapsto\sigma'$.

\begin{defn}
	A \emph{neighboring coloring pair} is a tuple $(G,\sigma,\tau)$ where $\sigma,\tau\in\Omega$ are colorings of $G$ (not necessarily proper) which differ on a single vertex, which we will denote throughout this paper by $v$. Where the context is clear, we will omit $G$ and refer to neighboring coloring pairs as $(\sigma,\tau)$.
\end{defn}

\begin{defn}
	For $\sigma,\tau\in\Omega$, the \emph{Hamming distance} $d(\sigma,\tau)$ is the number of vertices on which $\sigma,\tau$ differ.
\end{defn}

\begin{thm}[\cite{bubley1997path}]
	If the Markov chain $\sigma\mapsto\sigma'$ has a partial coupling $(\sigma,\tau)\mapsto(\sigma',\tau')$ defined for all neighboring coloring pairs $(\sigma,\tau)$ that $(1-\alpha)$-contracts for some $\alpha > 0$, then there exists a coupling defined for all pairs of colorings $(\sigma,\tau)\in\Omega^2$ which $(1-\alpha)$-contracts.
\end{thm}

\begin{remark}
	The reason we need to extend the state space from $\Omega^*$ to $\Omega$ is that in the context of path coupling, we want to extend the premetric of all neighboring colorings to the Hamming metric. But given two colorings $\sigma,\tau$ for which $d(\sigma,\tau) = \ell$, there does not necessarily exist a sequence of \emph{proper} colorings $\sigma = \sigma_0, \sigma_1, ..., \sigma_{\ell} = \tau$ for which $\sigma_i$ and $\sigma_{i+1}$ are neighboring for all $0\le i < \ell$. However, there certainly exist such sequences if we allow the colorings to be improper.

	This is a standard fix that comes up in typical applications of path coupling to sampling colorings. As noted by the authors in these applications, the stationary distribution for this Markov chain is still the uniform distribution over proper colorings. This is because if we start at a proper coloring, we only ever visit proper colorings, and if start at an improper coloring, we eventually reach a proper coloring, so the support of the stationary distribution consists only of proper colorings.
\end{remark}

Jerrum's $k\ge 2d$ bound \cite{jerrum1995very} and Vigoda's $k\ge(11/6)d$ bound \cite{vigoda2000improved} can both be proved via \emph{one-step} path couplings. Yet there is substantial evidence that multi-step couplings are far stronger than one-step couplings. As shown in \cite{kumar2001coupling}, there exist Markov chains for some sampling problems where one-step coupling analysis is insufficient. \cite{czumaj1999delayed} used multi-step coupling for a tighter analysis of a Markov chain for sampling random permutations, and the celebrated $k\ge(1+\epsilon)d$ result of \cite{hayes2003non} for $\Omega(\log n)$-degree graphs uses a multi-step coupling which is constructed by looking into future time steps.

There are also several other works that carried out a multi-step coupling analysis by looking at one-step coupling over multiple time steps \cite{dyer2001extension,dyer2002very,hayes2007variable} and obtained slight improvements over Jerrum's $k\ge 2d$ bound by terminating path coupling of the Glauber dynamics at a random stopping time. In the literature, this is known as \emph{variable-length coupling}, and this is the approach we take, but for the flip dynamics.


\begin{defn}[Definition 1 in \cite{hayes2007variable}]
	For every initial neighboring coloring pair $(\sigma^{(0)},\tau^{(0)})$, let $(\overbar{\sigma},\overbar{\tau},\Tstop)$ be a random variable where $\Tstop$ is a nonnegative integer, and $\overbar{\sigma},\overbar{\tau}$ are sequences of colorings $(\sigma^{(0)},...,\sigma^{(\Tstop)})\in\Omega^{\Tstop}$ and $(\tau^{(0)},...,\tau^{(\Tstop)})\in\Omega^{\Tstop}$ respectively. We say that $(\overbar{\sigma},\overbar{\tau},\Tstop)$ is a \emph{variable-length path coupling} if $\overbar{\sigma},\overbar{\tau}$ are faithful copies of the Markov chain in the following sense.

	For $(\sigma^{(0)},\tau^{(0)})$ and $t\ge 0$, define the random variables $\sigma_t, \tau_t$ via the following experiment: 1) sample $(\overbar{\sigma},\overbar{\tau},\Tstop)$, 2) if $t\le T$, define $\sigma_t = \sigma^{(t)}, \tau_t = \tau^{(t)}$, 3) if $t > \Tstop$, then sample $\sigma_t$ and $\tau_t$ from $P^{t-\Tstop}(\sigma^{(\Tstop)},\cdot)$ and $P^{t-\Tstop}(\tau^{(\Tstop)},\cdot)$ respectively.

	We say that $\overbar{\sigma}$ (resp. $\overbar{\tau}$) is a \emph{faithful copy} if for every neighboring coloring pair $(\sigma^{(0)},\tau^{(0)})$ and $t\ge 0$, $\sigma_t$ and $\tau_t$ defined above are distributed according to $P^t(\sigma^{(0)},\cdot)$ and $P^t(\tau^{(0)},\cdot)$ respectively.\label{def:varlength}
\end{defn}

Note that when $\Tstop$ is always equal to some fixed $T$, then this is just the usual notion of $T$-step path coupling. The following is an extension of the path coupling theorem of \cite{bubley1997path} to variable-length path couplings.

\begin{thm}[Corollary 4 of \cite{hayes2007variable}]
	For a variable-length path coupling $(\overbar{\sigma},\overbar{\tau},\Tstop)$, let $$\alpha\triangleq 1 - \max_{\sigma^{(0)},\tau^{(0)}}\E[d(\sigma^{(\Tstop)},\tau^{(\Tstop)})],\quad W \triangleq\max_{\sigma^{(0)},\tau^{(0)},t\le\Tstop} d(\sigma^{(t)},\tau^{(t)}), \quad\beta\triangleq \max_{\sigma^{(0)},\tau^{(0)}}\E[\Tstop].$$ If $\alpha > 0$, then the mixing time satisfies $\tau_{mix}\le 2\left\lceil 2\beta W/\alpha\right\rceil\cdot\left\lceil\ln(n)/\alpha\right\rceil.$\label{thm:hayesvigoda}
\end{thm}

\subsection{Vigoda's One-Step Coupling as an LP}
\label{subsec:introonestep}

In this section we review the one-step path coupling analysis from \cite{vigoda2000improved} and give an interpretation of Vigoda's argument as implicitly finding a feasible point for a linear program which turns out to be an optimizer. 

Fix a neighboring coloring pair $(G,\sigma,\tau)$. Note that the symmetric difference $D$ between the set of all alternating components $S_{\sigma}(x,c)$ in $\sigma$ and the set of all alternating components $S_{\tau}(x',c)$ in $\tau$ is precisely the alternating components $S_{\sigma}(u,\tau(v))$ and $S_{\sigma}(v,c)$ in $\sigma$ and the alternating components $S_{\tau}(u,\sigma(v))$ and $S_{\tau}(v,c)$ in $\tau$, for all $u\in N(v)$ and colors $c$ appearing in the neighborhood of $v$. All other alternating components are shared between $\sigma$ and $\tau$, so for those, it's enough to use the identity coupling. Note that for colors $c\neq\sigma(v),\tau(v)$ not appearing in $N(v)$, the identity coupling then matches the flip of $S_{\sigma}(v,c)$ to that of $S_{\tau}(v,c)$ so that the two colorings of $G$ become identical.

So the main concern is how to couple the flips of components in $D$. Let $\delta_c$ denote the number times color $c$ appears in the neighborhood of $v$. We can decompose $D$ into $\cup_{c:\delta_c > 0} D_c$, where the sets $D_c$ are defined as follows:

\begin{defn}
	$D_c$ is set of alternating components consisting of $S_{\sigma}(v,c),S_{\tau}(v,c)$, and all $S_{\sigma}(u,\tau(v))$ and $S_{\tau}(u,\sigma(v))$ for all $c$-colored neighbors $u$ of $v$.\label{def:Dc}
\end{defn}

Informally, $D_c$ is the subset of $D$ that involves the color $c$. It's easy to see that for $c\neq \sigma(v)$, $$S_{\sigma}(v,c) = \left(\bigcup_{u\in N(v): \sigma(u) = c}S_{\tau}(u,\sigma(v))\right)\cup\{v\},$$ and when $c = \sigma(v)$, $S_{\sigma}(v,c), S_{\tau}(u,\sigma(v))= \emptyset$. Likewise we have that for $c\neq\tau(v)$, $$S_{\tau}(v,c) = \left(\bigcup_{u\in N(v): \sigma(u) = c}S_{\sigma}(u,\tau(v))\right)\cup\{v\},$$ and when $c = \tau(v)$, $S_{\tau}(v,c), S_{\sigma}(u,\tau(v)) = \emptyset$. The sets $D_c$ are disjoint except possibly the pair $D_{\sigma(v)},D_{\tau(v)}$, as these both contain $(\sigma(v),\tau(v))$-colored alternating components, though we defer this point to later. Another subtlety is that there may exist multiple neighbors $u_1, \cdots, u_m\in N(v)$ which are colored $c$ but which satisfy $S_{\tau}(u_1,\sigma(v)) = \cdots = S_{\tau}(u_m,\sigma(v))$; to guarantee that the flip of each component is considered exactly once, redefine $S_{\tau}(u_i,\sigma(v)) = \emptyset$ for all $1 < i\le m$. Handle the components $S_{\sigma}(u_i,\tau(v))$ analogously.

In \cite{vigoda2000improved}, Vigoda couples flips of alternating components within $D_c$ as follows. First we require some notation. For $c$ such that $\delta_c > 0$, define $A_c := |S_{\sigma}(v,c)|$, $B_c := |S_{\tau}(v,c)|$, $\vec{a}^c := (|S_{\tau}(u,\sigma(v))|)_{u\in N(v): \sigma(u) = c}$, and $\vec{b}^c := (|S_{\sigma}(u,\tau(v))|)_{u\in N(v): \sigma(u) = c}$ for every color $c$ in the neighborhood of $v$. Also define $a^c_{max} = \max_i a^c_i$ and denote the maximizing $i$ by $i^c_{max}$. Likewise define $b^c_{max} = \max_j b^c_j$ and denote the maximizing $j$ by $j^c_{max}$. When it is clear from context that we are just focusing on a generic color $c$, we will refer to these as $A,B,\vec{a},\vec{b}, a_{max}, i_{max},b_{max},j_{max}$, and we will refer to the $c$-colored neighbors of $v$ as $u_1, \cdots, u_{\delta_c}$ and the corresponding entries in $\vec{a},\vec{b}$ as $a_1,\cdots, a_{\delta_c}$ and $b_1,\cdots, b_{\delta_c}$. Naively we have the bounds \begin{equation}1 + a_{max}\le A \le 1+ \sum_i a_i, \ \ \ \ \ 1 + b_{max}\le B\le 1 + \sum_i b_i.\label{eq:trivABbound}\end{equation}

Note that $S_{\sigma}(v,c)$ and $S_{\tau}(v,c)$ can be quite different but $S_{\sigma}(v,c) \supset S_{\tau}(u_i,\sigma(v))$ so it is easier to understand the overlap between these two components. Among all choices of $i$, this overlap is maximized for $i = i_{max}$, and the idea of Vigoda's coupling is thus to greedily couple the flips of the biggest components, i.e. $S_{\sigma}(v,c), S_{\tau}(v,c)$, to the flips of the next biggest components, i.e. $S_{\tau}(u_{i_{max}},\sigma(v)), S_{\sigma}(u_{j_{max}},\tau(v))$, and then to couple together as closely as possible the flips of $S_{\sigma}(u_i,\tau(v))$ and $S_{\tau}(u_i,\sigma(v))$ for each $i\in[\delta_c]$. Formally, assuming $p_1 \geq p_2 \geq \cdots $ we have:

\begin{enumerate}
	\item Flip $S_{\sigma}(v,c)$ and $S_{\tau}(u_{i_{max}},\sigma(v))$ together with probability $p_A$.
	\item Flip $S_{\tau}(v,c)$ and $S_{\sigma}(u_{j_{max}},\tau(v))$ together with probability $p_B$.
	\item For $i\in[\delta_c]$, define \begin{equation}
		q_i = \begin{cases}
			p_{a_i} - p_A & \text{if $i = i_{max}$} \\
			p_{a_i} & \text{otherwise}
		\end{cases}
		\label{eq:qi}
	\end{equation}

	\begin{equation}
		q'_i = \begin{cases}
			p_{b_i} - p_B & \text{if $i = j_{max}$} \\
			p_{b_i} & \text{otherwise}
		\end{cases}
		\label{eq:qprimei} 
	\end{equation}

	Note that $q_i$ and $q'_i$ are the remaining probability associated to flips $S_{\tau}(u_i,\sigma(v))$ and $S_{\sigma}(u_i,\tau(v))$ respectively.
	
	\begin{enumerate}
		\item Flip $S_{\tau}(u_i,\sigma(v))$ and $S_{\sigma}(u_i,\tau(v))$ together with probability $\min(q_i,q'_i)$
		\item Flip only $S_{\tau}(u_i,\sigma(v))$ together with probability $q_i - \min(q_i,q'_i)$
		\item Flip only $S_{\sigma}(u_i,\tau(v))$ together with probability $q'_i - \min(q_i,q'_i)$
	\end{enumerate}
\end{enumerate}

Coupled moves 1) and 2) change the Hamming distance by at most $A - a_{max} - 1$ and $B - b_{max} - 1$ respectively (with equality, for example, if $G$ is a tree rooted at $v$). For any given $i\in[\delta_c]$, coupled move 3a) changes the Hamming distance by $a_i + b_i - 1$, while coupled moves 3b) and 3c) change the Hamming distance by $a_i$ and $b_i$ respectively. For $A,B,\vec{a},\vec{b}$, define

\begin{equation}
	H(A,B,\vec{a},\vec{b}) = (A - a_{max} - 1)p_A + (B - b_{max} - 1)p_B + \sum_i f(u_i), \label{eq:H}
\end{equation}

where \begin{equation}
	f(u_i) = a_iq_i + b_iq'_i - \min(q_i,q'_i)\label{eq:fi}
\end{equation}

The above discussion implies that for $c\neq\sigma(v),\tau(v)$ appearing in the neighborhood of $v$, \begin{equation}\E[d(\sigma',\tau') - 1\vert X_c]\le H(A_c,B_c,\vec{a}^c,\vec{b}^c),\label{eq:mainvigodaineqcomponent}\end{equation} where $X_c$ is the random event that the coupling flips components in $D_c$ in both colorings. For $c$ not appearing in the neighborhood of $v$, the Hamming distance will not change if alternating components containing the color $c$ are flipped in both colorings, as the coupling is the identity on these components, except if $v$ is flipped to $c$ in both colorings, in which case the Hamming distance decreases by 1.

Lastly, we review how the case of $c = \sigma(v),\tau(v)$ and $D_{\sigma(v)}\cup D_{\tau(v)}\neq\emptyset$ is handled in \cite{vigoda2000improved}. This is the main place where one needs to be careful about the fact that neighboring coloring pairs $\sigma,\tau$ need not be proper.

\begin{remark}
	When $c =\sigma(v),\tau(v)$, we must make sure not to double count flips, as it is possible that $D_{\sigma(v)}$ and $D_{\tau(v)}$ share alternating components. In this remark, suppose $D_{\sigma(v)}\cap D_{\tau}(v)\neq\emptyset$. This can only happen if there exist $x_i,y_j\in N(v)$ colored $\sigma(v),\tau(v)$ respectively for which $S_{\sigma}(v,\tau(v)) = S_{\sigma}(x_i,\tau(v))$ and $S_{\tau}(v,\sigma(v)) = S_{\tau}(x_i,\sigma(v))$. To avoid double counting, Vigoda sets $S_{\sigma}(v,\tau(v)) = S_{\tau}(y_j,\sigma(v)) = \emptyset$ in this case. The bound \eqref{eq:mainvigodaineqcomponent} then holds for both $c = \sigma(v),\tau(v)$. The only difference is that some values among $A_c,B_c$ and the entries of $\vec{a}^c,\vec{b}^c$ will be zero, in which case we take $p_0 = 0$.

	Specifically, for $c = \tau(v)$, we have $A_c = 0$, $B_c = b^c_{max} = 0$, and at least one $a^c_j$ is zero, namely the one corresponding to the component $S_{\tau}(y_j,\sigma(v)) = S_{\tau}(v,\sigma(v))$. In this case one can check that $H(A_c,B_c,\vec{a}^c,\vec{b}^c) = \sum a^c_jp_{a^c_j}$, and provided $\alpha p_{\alpha}\le 1$ for all $\alpha$, this is at most $\delta_c - 1$.

	For $c = \sigma(v)$, we have $A_c = 0$, $a^c_i = 0$ for all $i$, and $B_c = \sum_j b^c_j$. Let $j^*$ be the index of the unique neighbor $u_{j^*}$ of $v$ for which $S_{\tau}(v,\sigma(v)) = S_{\sigma}(u_{j^*},\sigma(v))$. Then because $S_{\sigma}(u_{j^*},\sigma(v))$ contains $v$, we need to modify the definition of $b^c_{max}$. Let $b^c_{max} = \max_j(b^c_j - \mathbb{I}_{j = j^*})$ and denote the maximizing $j$ by $j^c_{max}$. Then the lower bound on $B_c$ in \eqref{eq:trivABbound} still holds, and \eqref{eq:mainvigodaineqcomponent} still holds. Moreover, if $\delta_c = 1$, then $\E[d(\sigma',\tau') - 1|X_{\sigma(v)}] = H(A_c,B_c,\vec{a}^c,\vec{b}^c) = 0$.

	\label{remark:specialcase}
\end{remark}

Henceforth, we will refer to the coupling defined above as the \emph{greedy coupling}. We can conclude the following, implicit in \cite{vigoda2000improved}:

\begin{lem}
	Let $(\sigma,\tau)\mapsto(\sigma',\tau')$ be the greedy coupling. Then \begin{equation}\E[d(\sigma',\tau') - 1] \le\frac{1}{nk}\left(-|\{c:\delta_c = 0\}| + \sum_{c: \delta_c\neq 0}H(A_c,B_c,\vec{a}^c,\vec{b}^c)\right).\label{eq:mainvigodaineq}\end{equation}
\end{lem}

The function $H$ implicitly depends on the choice of flip probabilities $\{p_{\alpha}\}$, while $A_c,B_c,\vec{a}^c,\vec{b}^c$ depend on $(G,\sigma,\tau)$. The remaining analysis in \cite{vigoda2000improved} once \eqref{eq:mainvigodaineq} has been deduced essentially boils down to picking $\{p_{\alpha}\}$\footnote{Note that one must set $p_1 = 1$ because otherwise, by rescaling all flip probabilities by a factor of $1/p_1$, the mixing time simply scales by a factor of $1/p_1$.} to minimize the right-hand side of \eqref{eq:mainvigodaineq} over all graphs $G$ of max-degree $d$ and $k$-colorings $\sigma,\tau$ of $G$. The following gives terminology for quantifying over all such $(G,\sigma,\tau)$.

\begin{defn}
	$(A,B,\vec{a},\vec{b})$ is \emph{realizable} if there exists a neighboring coloring pair $(G,\sigma,\tau)$ and color $c$ such that $(A,B,\vec{a},\vec{b}) = (A_c,B_c,\vec{a}^c,\vec{b}^c)$.
\end{defn}

Vigoda's remaining analysis can thus be interpreted as solving the following linear program.

\begin{lp}
	For variables $\{p_{\alpha}\}_{\alpha\in\N}$ and $\lambda$, minimize $\lambda$ subject to the following constraints: $0\le p_{\alpha}\le p_{\alpha-1}\le p_1 = 1$ for all $\alpha\ge 2$, and for all realizable $(A,B,\vec{a},\vec{b})$, define a constraint \begin{equation}
		 H(A,B,\vec{a},\vec{b})\le -1 + \lambda\cdot m,
	\end{equation} where $m$ is the number of entries in $\vec{a}$. Note that we can implement the $\min(\cdot)$ terms in the definition of $f(\cdot)$ by introducing the appropriate auxiliary variables so that the above is still a linear program.
	\label{def:prelp}
\end{lp}

There are three minor issues with this linear program: $(a)$ the linear program has an infinite number of variables, $(b)$ it has an infinite number of constraints, and $(c)$ given $\vec{a},\vec{b}$, it is not immediately obvious how to enumerate all $A,B$ for which $(A,B,\vec{a},\vec{b})$ is realizable.

He handles $(a)$ by restricting to flips of components of size at most $N_{max}$, i.e. by fixing some small constant $N_{max}$ and insisting that \begin{equation}p_{\alpha} = 0 \ \forall \alpha > N_{max}.\label{eq:NMax}\end{equation}

He handles $(b)$ by shrinking the feasible region as follows. Define $$g(u_i) = a_ip_{a_i} + b_ip_{b_i} - \min(p_{a_i},p_{b_i})$$ and note the following.

\begin{lem}
$H(A,B,\vec{a},\vec{b})\le (A - 2)p_A + (B - 2)p_B + \sum_i g(u_i)$.\label{lem:crude}
\end{lem}

\begin{proof}
For $i\neq i_{max},j_{max}$, we have that $g(u_i) = f(u_i)$. If $i_{max} = j_{max} = i$, then $$f(u_i)\le g(u_i) + p_A(-a_{max} + 1) + p_B(-b_{max} + 1),$$ and if $i_{max}\neq j_{max}$, then $$f(u_{i_{max}}) + f(u_{j_{max}}) \le g(u_{i_{max}}) + g(u_{j_{max}}) + p_A(-a_{max} + 1) + p_B(-b_{max} + 1).$$ Now we have $$\sum_i f(u_i)\le \sum_i g(u_i) + p_A(-a_{max} + 1) + p_B(-b_{max} + 1),$$ from which the desired bound follows after noting that $a_{max},b_{max}\ge 1$.\end{proof}

Whereas $f(u_i)$ are linear functions of $p_{a_i},p_{b_i},p_A,p_B$, $g(u_i)$ is simply a linear function of $p_{a_i},p_{b_i}$. So we can pick some $m^*$ (\cite{vigoda2000improved} picks $m^* = 3$) and replace the infinitely many constraints for which $m\ge m^*$ in Linear Program~\ref{def:prelp} with finitely many constraints to optimize the upper bound on $H(A,B,\vec{a},\vec{b})$ in Lemma~\ref{lem:crude}.

Finally, Vigoda implicitly handles $(c)$ above by relaxing the requirement of realizability as follows. To cover all constraints corresponding to $c\neq\sigma(v),\tau(v)$, include \begin{equation}H(A,B,\vec{a},\vec{b})\le -1 + \lambda\cdot m\label{eq:mainconstraint}\end{equation} for all $m < m^*$ and $(A,B,\vec{a},\vec{b})$ for which $\vec{a},\vec{b}\in\{0,1,\cdots,N_{max}\}^m\backslash\{(0,0,\cdots 0)\}$ and $A,B$ satisfy \eqref{eq:trivABbound}. To cover all constraints corresponding to $c=\sigma(v)$, include the constraint \begin{equation}(B - b_m)p_B + \sum^{m-1}_{i=1}b_ip_{b_i}\le -1 + \lambda\cdot m\label{eq:Hforsigmav}\end{equation} for all $2\le m < m^*$, $0\le b_1\le\cdots\le b_m\le N_{max}$ where $b_m > 0$, and $B = \sum_i b_i$. Indeed we know by Remark~\ref{remark:specialcase} that for $c = \sigma(v)$, any realizable $(A_c,B_c,\vec{a}^c,\vec{b}^c)$ satisfies $H(A_c,B_c,\vec{a}^c,\vec{b}^c) = 0$ if $\delta_c = 1$, and satisfies $A_c = 0$, $\vec{a}^c = (0,...,0)$, $B_c = \sum_i b^c_i$, and $$H(A_c,B_c,\vec{a}^c,\vec{b}^c) = (B_c - b^c_{max} - 1)p_{B_c} + \sum_{i\neq j_{max}} b^c_ip_{b^c_i} + b_{j_{max}}(p_{b^c_{j_{max}}} - p_{B_c}) \le (B_c - b^c_m)p_{B_c} + \sum b^c_ip_{b^c_i}$$ if $\delta_c > 1$. So this relaxation covers all constraints corresponding to $c = \sigma(v)$. And to cover all constraints corresponding to $c = \tau(v)$, merely include the constraint \begin{equation}\alpha\cdot p_{\alpha}\le 1 \ \forall\alpha\in\N.\label{eq:pap}\end{equation} As observed in Remark~\ref{remark:specialcase}, when $c = \tau(v)$, \eqref{eq:pap} ensures that $H(A_c,B_c,\vec{a}^c,\vec{b}^c)\le\delta_c - 1$, and for any $\lambda > 1$ (corresponding to $k > d$, which is the regime we are interested in to begin with), we automatically have that $\delta_c - 1 < -1 + \lambda\cdot\delta_c$.

Concretely, we have the following linear program.

\begin{lp}
	Fix some $N_{max}\ge 1$ and $m^*\ge 2$. For variables $\{p_{\alpha}\}_{\alpha\in\N}$ and $\lambda$, and dummy variables $x, y$, minimize $\lambda$ subject to the following constraints: $0\le p_{\alpha}\le p_{\alpha-1}\le p_1 = 1$ for all $\alpha\ge 2$, constraint \eqref{eq:NMax}, constraint \eqref{eq:mainconstraint} for all $m < m^*$ and $(A,B,\vec{a},\vec{b})$ for which $\vec{a},\vec{b}\in\{0,1,\cdots,N_{max}\}^m\backslash\{(0,0,\cdots 0)\}$ and $A,B$ satisfy \eqref{eq:trivABbound}, constraint \eqref{eq:Hforsigmav} for all $2\le m < m^*$ and $0\le b_1\le\cdots\le b_m\le N_{max}$ where $b_m > 0$ and $B = \sum_i b_i$, constraint \eqref{eq:pap}, and constraints \begin{equation*}x\ge (A - 2)p_A\end{equation*} \begin{equation*}y\ge a\cdot p_a + b\cdot p_b - \min(p_a,p_b)\end{equation*} \begin{equation}-1 + \lambda\cdot m^* \ge 2x + m^*\cdot y\label{eq:approxconstraint}\end{equation} for every $A,a,b$ satisfying $0\le A\le 1+N_{max}$ and $0\le a < b\le N_{max}$.
	\label{def:lp}
\end{lp}

\begin{remark}
	Note that we only add in constraints for the upper bound of Lemma~\ref{lem:crude} in the case of $m = m^*$ (constraint \eqref{eq:approxconstraint}), but this is because the constraint for $m = m^*$ immediately implies the corresponding constraints for $m > m^*$.
\end{remark}

From the above discussion, the constraints of Linear Program~\ref{def:lp} are stronger than those of Linear Program~\ref{def:prelp}. It follows that

\begin{obs}
	The objective value of Linear Program~\ref{def:prelp} is at most that of Linear Program~\ref{def:lp}. Moreover, the objective value of Linear Program~\ref{def:lp} is non-increasing in $m^*$.\label{obs:comparelambdas}
\end{obs}

The following gives one direction of the connection between the analysis of greedy one-step coupling and the objective value of Linear Program~\ref{def:prelp} in the following lemma.

\begin{lem}\label{lem:contractive}
Let $\lambda_2^*$ be the objective value of Linear Program~\ref{def:lp}.
If $\lambda_2^*\ge 1$ and $k > \lambda_2^*d$, then there exist flip probabilities $\{p_{\alpha}\}_{\alpha\in\N}$ for which $\E[d(\sigma',\tau') - 1] < 0$ for all such $G,\sigma,\tau$, where $(\sigma,\tau)\mapsto(\sigma',\tau')$ is the greedy coupling.\label{lem:vigodaoptimal}
\end{lem}

\begin{proof}
	Let $\{p_{\alpha}\}_{\alpha\in\N}$ be flip probabilities achieving objective value $\lambda_2^*$. By \eqref{eq:mainvigodaineq} and Lemma~\ref{lem:crude}, we have that $$nk\cdot \E[d(\sigma',\tau'] - 1]\le -|\{c: \delta_c = 0\}| + \sum_{c:\delta_c\neq 0} H(A_c,B_c,\vec{a}^c,\vec{b}^c)\le -k + \lambda^*\cdot d < 0$$ as desired.
\end{proof}

Later we will prove the converse that if $k < \lambda_2^*d$, then then for any flip probabilities $\{p_{\alpha}\}_{\alpha\in\N}$ there exists a neighboring coloring pair $(G,\sigma,\tau)$, such that $\E[d(\sigma',\tau')] > 1$. This is not obvious. It is \emph{a priori} unclear how to conclude this even about Linear Program~\ref{def:prelp} because its constraint set is infinite. It is even less clear how to conclude this about Linear Program~\ref{def:lp} because it is a relaxation of Linear Program~\ref{def:prelp}, so even the objective values of the linear programs need not agree. We address these issues in Section~\ref{subsec:tightconstraints} and show that in fact any optimizer of Linear Program~\ref{def:lp} gives an optimal one-step coupling for the flip dynamics.


In \cite{vigoda2000improved}, Vigoda shows that for $m^* = 3$, $N_{max} = 6$, and the following flip probabilities, Linear Program~\ref{def:lp} attains a value of $11/6$: \begin{equation}
	p_1 = 1, p_2 = 13/42, p_3 = 1/6, p_4 = 2/21, p_5 = 1/21, p_6 = 1/84, p_{\alpha} = 0 \ \forall \alpha\ge 7.\label{eq:vigodasprobs}
\end{equation}

Not only is this a feasible solution to Linear Program~\ref{def:lp}, but it turns out to be optimal. We prove in Section~\ref{sec:tightexamples} that increasing $m^*, N_{max}$ and changing $\{p_{\alpha}\}$ will not yield a better value. In fact, we will see that this remains true even if we restrict our attention to a small subset of the constraints, which will guide us to a family of examples that are tight for any one-step path coupling of the flip dynamics and which will motivate our strategy for breaking Vigoda's $k\ge (11/6)d$ barrier.


\section{A Closer Look at Vigoda's Linear Program}
\label{sec:tightexamples}

\subsection{A Tight Collection of Constraints}
\label{subsec:tightconstraints}

For purposes we explain later in Remark~\ref{remark:altprobs}, in this section we will consider the following flip probabilities rather than those chosen in \cite{vigoda2000improved}. We emphasize that these flip probabilities are chosen for the purposes of the discussion in this section only and are not the flip probabilities we will use in our final analysis.

\begin{equation}
	p_1 = 1, p_2 = \frac{463}{1500}, p_3 = \frac{1}{6}, p_4 = \frac{287}{3000}, p_5 = \frac{29}{600}, p_6 = \frac{71}{3000}, p_{\alpha} = 0 \ \forall \alpha\ge 7.\label{eq:ourprobs1}
\end{equation}
	
One can check that this choice is also feasible for Linear Program~\ref{def:lp} and attains a value of $11/6$ even for $m^* = 3$. What is important about this choice for our purposes is which constraints of Linear Program~\ref{def:lp} they make tight and which are slack:

\begin{obs}
	Let $N_{max}\ge 6$ in Linear Program~\ref{def:lp} and $m^* = 3$. \eqref{eq:pap} is tight under the assignment \eqref{eq:ourprobs1} only for $\alpha = 1$. Among the constraints of the form \eqref{eq:mainconstraint}, the only constraints that are tight under the assignment \eqref{eq:ourprobs1} are those for which either 1) $m = 1$ and either $a_1 = 1$ and $2\le b_1\le 4$ or $b_1 = 1$ and $2\le a_1\le 4$, 2) $m = 2$ and either $(a_1,a_2,b_1,b_2) = (1,1,3,3)$, $A = 3$, and $6\le B\le 7$ or $(a_1,a_2,b_1,b_2) = (3,3,1,1)$, $6\le A\le 7$, and $B = 3$.

	Any other constraints of the form \eqref{eq:mainconstraint} that do not meet these conditions, and all constraints of the form \eqref{eq:Hforsigmav} and \eqref{eq:approxconstraint}, are not tight under the assignment \eqref{eq:ourprobs1}.
	\label{obs:slack}
\end{obs}

This can be verified numerically. We give a proof in Appendix~\ref{app:obsproof} for completeness. 

\begin{remark}
	Under Vigoda's choice of flip probabilities \eqref{eq:vigodasprobs}, the constraint \eqref{eq:mainconstraint} for $m = 2$ and $(a_1,a_2,b_1,b_2,A,B) = (1,1,2,2,3,5)$ is tight in addition to the constraints from Observation~\ref{obs:slack}. But as Observation~\ref{obs:slack} demonstrates, there are other choices of flip probabilities, e.g. \eqref{eq:ourprobs1}, for which these additional constraints have positive slack so there are more free directions to move in while maintaining optimality. \label{remark:altprobs}
\end{remark}

If we restrict our attention to the constraints with zero slack in Observation~\ref{obs:slack}, we obtain the following linear program:

\begin{lp}
	For nonnegative variables $\{p_{\alpha}\}_{\alpha\in\N}$ with $0\le p_{\alpha}\le p_{\alpha-1}\le p_1 = 1$ for all $\alpha\ge 2$, minimize $\lambda$ subject to: 

	$$p_1 + p_2 - 2p_3 - \min(p_1-p_2,p_2-p_3) \le -1 +\lambda$$ $$p_1 - p_2 + 3p_3 - 3p_4 - \min(p_1 - p_2,p_3 - p_4) \le -1 + \lambda$$ $$p_1 - p_2 + 4p_4 - 4p_5 - \min(p_1-p_2,p_4-p_5)\le -1 + \lambda$$ $$2p_1 + 5p_3 - \min(p_1 - p_3,p_3 - p_6)\le -1 + 2\lambda.$$ $$2p_1 + 5p_3 - \min(p_1 - p_3,p_3 - p_7)\le -1 + 2\lambda.$$\label{def:mostbasictight}
\end{lp}

\begin{cor}
	The objective values of Linear Program~\ref{def:prelp}, Linear Program~\ref{def:lp} with $N_{max}\ge 7$ and $m^* = 3$, and Linear Program~\ref{def:mostbasictight} are all equal to 11/6.\label{cor:mostbasictightcor}
\end{cor}

\begin{proof}
	Denote these objective values by $\lambda^*_1, \lambda^*_2, \lambda^*_3$. By complementary slackness, Observation~\ref{obs:slack} implies $\lambda^*_2 = \lambda^*_3$ (this can also be verified numerically). But note that $$\lambda^*_3\le\lambda^*_1\le\lambda^*_2.$$ The second inequality follows by Observation~\ref{obs:comparelambdas}. For the first inequality, note that although the tuple $(A,B,\vec{a},\vec{b}) = (6,3,(3,3),(1,1))$ corresponding to the penultimate constraint of Linear Program~\ref{def:mostbasictight} is not realizable, the last two constraints of Linear Program~\ref{def:mostbasictight} are symmetric with respect to flipping $p_6$ and $p_7$ and are also the only constraints involving $p_6$ and $p_7$, so $\lambda^*_3$ is equal to the objective value of the linear program obtained by removing the penultimate constraint of Linear Program~\ref{def:mostbasictight}. Finally, the tuples $(a+1,2,(a),(1))$ for $a = 2,3,4$ are all realizable, as is the tuple $(7,3,(3,3),(1,1))$, so the constraint set of Linear Program~\ref{def:mostbasictight} is a strict subset of that of Linear Program~\ref{def:lp}, completing the proof.
\end{proof}

\subsection{A Worst-Case Family of Examples for One-Step Coupling}

As discussed after Lemma~\ref{lem:contractive}, it is not immediately obvious that a minimizer of either Linear Program~\ref{def:prelp} or Linear Program~\ref{def:lp} characterizes the optimal one-step coupling analysis for the flip dynamics. One way to interpret Corollary~\ref{cor:mostbasictightcor} is as a resolution to this question via the following lemma, which exhibits a family $\mathcal{C}$ of just \emph{four} neighboring coloring pairs $(G,\sigma,\tau)$ for which no one-step coupling, greedy or otherwise, simultaneously contracts with respect to the Hamming metric for $k < (11/6)d$ for all $(G,\sigma,\tau)\in\mathcal{C}$. To clarify, this lemma is not used in the proof of our main result, but provides intuition for the limitations of one-step coupling and motivates our approach for circumventing them.

\begin{figure}[ht]
	\centering
	\subcaptionbox{$G_1$}{\includegraphics[width=0.45\textwidth]{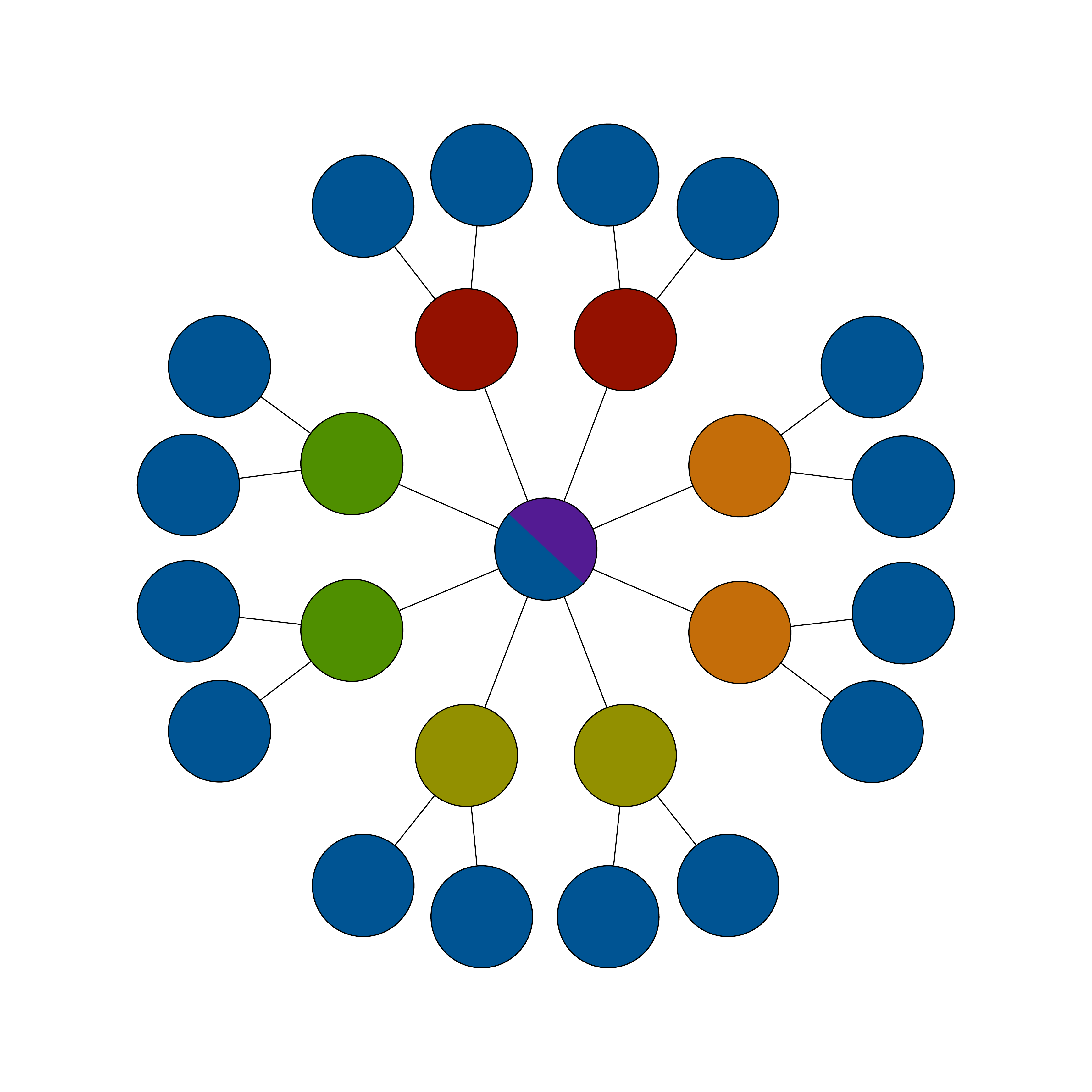}}
	\subcaptionbox{$G_2$}{\includegraphics[width=0.45\textwidth]{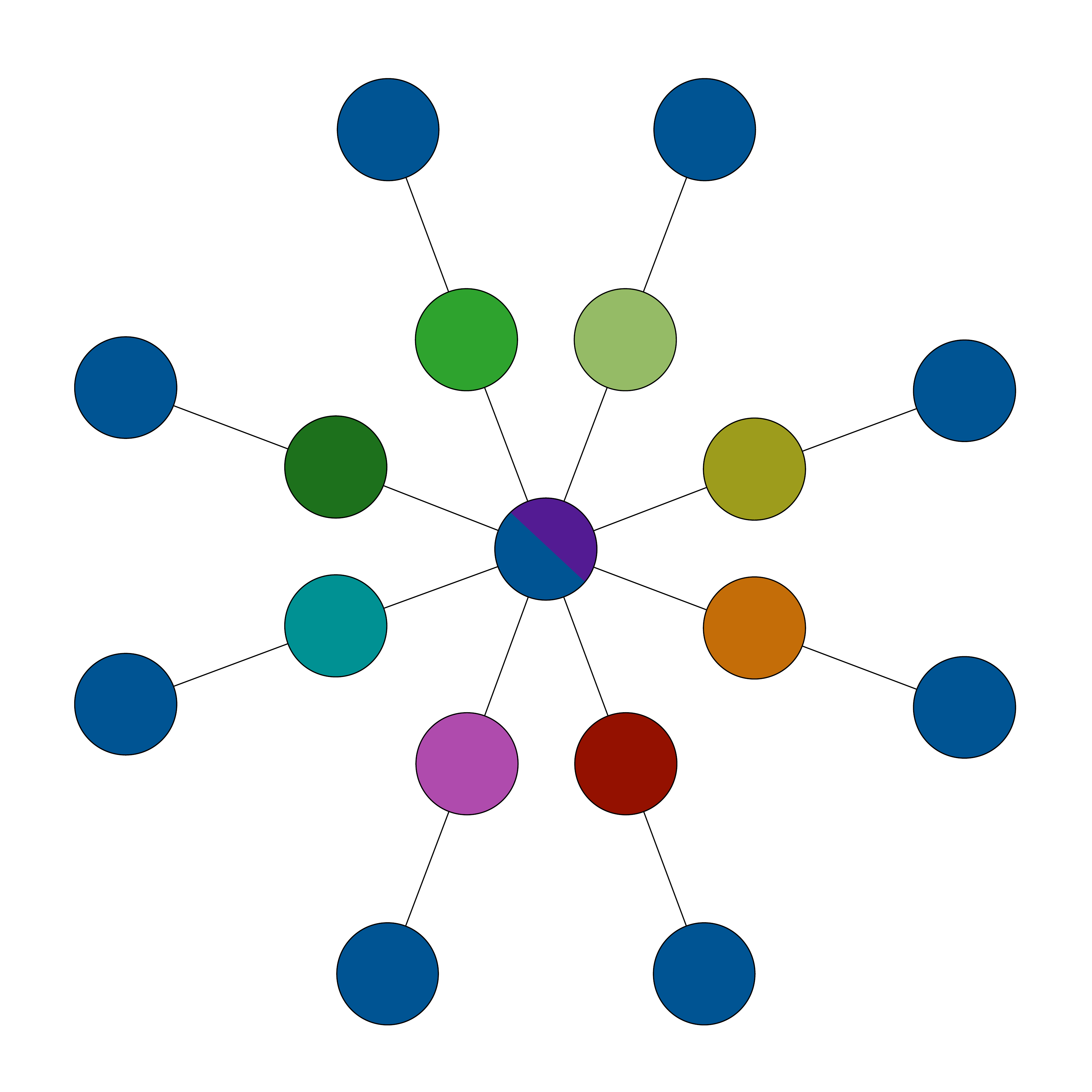}}
	\caption{Examples of graphs defined in Construction~\ref{constr}, where $\sigma(v) = \text{blue}$ and $\tau(v) = \text{purple}$}
	\label{fig:allGs}
\end{figure}

\begin{constr}
	Let $d$ be even. Let $G_1$ be the tree of height two rooted at a vertex $v$ with exactly $d$ children $u_1,...,u_d$ such that each $u_i$ has exactly two children $w^i_1$ and $w^i_2$. In colorings $\sigma_{1},\tau_{1}$, assign $u_{2j-1}$ and $u_{2j}$ the color $c_j$ for $j = 1,...,d$, and assign all $w^i_{\ell}$ the color $\sigma_1(v)\neq\tau_1(v)$

	For $a = 2,3,4$, let $G_a$ be the tree of height $a + 1$ rooted at vertex $v$ with $d$ path graphs, each consisting of $a$ other vertices, attached to $v$. Let $u_1,...,u_d$ be the neighbors of $v$. In colorings $\sigma_a,\tau_a$, assign each $u_i$ with a distinct color $c_i$, and assign all descendants of $u_i$ which are an odd distance away from $v$ with $c_i$. Assign all remaining descendants of $v$ with the color $\sigma_a(v)\neq\tau_a(v)$.

	Let $\mathcal{C}^* = \{(G_i,\sigma_i,\tau_i)\}_{1\le i\le 4}$ (see Figure~\ref{fig:allGs}).\label{constr}
\end{constr}

\begin{lem}\label{lem:tight}
	If $k < (11/6)d$, there exists no choice of flip probabilities $\{p_{\alpha}\}$ and one-step coupling $(\sigma,\tau)\mapsto(\sigma',\tau')$ for which $\E[d(\sigma',\tau') - 1] < 0$ for all $(G,\sigma,\tau)\in\mathcal{C}^*$, where $\mathcal{C}^*$ is defined in Construction~\ref{constr}.\label{lem:Cfamily}
\end{lem}

\begin{proof}
	We first show there exists no choice of flip probabilities for which greedy coupling contracts for all of $\mathcal{C}^*$. Let $\lambda = k/d$, and suppose to the contrary that $1\le \lambda < 11/6$ and yet there exists a set of flip probabilities $\{p_{\alpha}\}$ for which all pairs of colorings in $\mathcal{C}^*$ contracted in distance. For $a = 2,3,4$, the expected change in distance for $G_a$ is $$d\cdot H(a+1,2,(a),(1)) = d\left(p_1 - p_2 + a(p_a - p_{a+1}) - \min(p_1 - p_2,p_a - p_{a+1})\right) < 0 \le d(-1 + \lambda).$$ The expected change in distance for $G_1$ is $$(d/2)\cdot H(7,3,(3,3),(1,1)) = (d/2)\left(2p_1 + 5p_3 - \min(p_1 - p_3,p_3-p_7)\right) < 0 \le (d/2)(-1 + 2\lambda).$$ But this would imply that under this choice of $\{p_{\alpha}\}$, the linear program in Definition~\ref{def:mostbasictight} achieves a value of $\lambda < 11/6$, contradicting Corollary~\ref{cor:mostbasictightcor}.

	Finally, it is straightforward to check that no one-step coupling can do better than the greedy coupling. This is clear for $G_a$ with $a = 2,3,4$. Indeed, certainly for any component not in some $D_c$ for color $c$ in the neighborhood of $v$, the coupling should just be the identity. Now for any neighbor $u$ of $v$ with color $c$, suppose a nonzero amount of probability mass $p$ for the flip of $S_{\tau}(u,\sigma(v))$ is matched in the optimal one-step coupling to the flip of a component other than $S_{\sigma}(v,c)$. The expected change in distance conditioned on this pair of components being chosen in the coupling is strictly greater than the expected change if that mass $p$ were instead reallocated to the empty flip in $\sigma$, contradicting optimality. By symmetry we can show that the flip of $S_{\sigma}(u_i,\tau(v))$ is coupled only to the empty flip in $\tau$ and the flip of $S_{\tau}(u_i,c)$. Finally, if not all of the probability mass for the flip of $S_{\sigma}(v,c)$ is matched to the flip of $S_{\tau}(u,\sigma(v))$, then we can strictly improve the coupling by reallocating that mass to $S_{\tau}(u,\sigma(v))$.

	A similar argument shows that the optimal one-step coupling for $G_1$ is the greedy coupling.
\end{proof}

In other words, not only is it impossible for any one-step coupling analysis of the flip dynamics to cross Vigoda's 11/6 barrier for general graphs, but it's impossible even for the family of four \emph{tree} graphs defined in Lemma~\ref{lem:Cfamily}.

\subsection{Modifying the LP}
\label{subsec:modifyinglp}

For a color $c$, the condition that $(A_c,B_c,\vec{a}^c,\vec{b}^c)$ for a neighboring coloring pair $(G,\sigma,\tau)$ correspond to a constraint of Linear Program~\ref{def:mostbasictight} is a very stringent condition on $(G,\sigma,\tau)$. The hope is that for a suitable notion of ``typical,'' this condition holds for few colors $c$ for a ``typical'' neighboring coloring pair. Before exploring this line of thought, it will be convenient to give a name to neighboring coloring pairs with this condition.

\begin{defn}
Given a neighboring coloring pair $(G,\sigma,\tau)$ and a color $c$ appearing in the neighborhood of $v$, then the pair $\sigma,\tau$ is in the state \begin{enumerate}
	\item \Sing{c} if $\delta_c = 1$ and $c\neq\sigma(v),\tau(v)$,
	\item \Bad{c} if $\delta_c = 2$ and $(A_c,B_c,\vec{a}^c,\vec{b}^c)$ is either $(7,3,(3,3),(1,1))$ or $(3,7,(1,1),(3,3))$ (see Figure~\ref{fig:treeexamples}).
	\item \Good{c} otherwise.
\end{enumerate} Moreover, define $N_{sing}(\sigma,\tau)$, $N_{bad}(\sigma,\tau)$, and $N_{good}(\sigma,\tau)$ to be the number of $c$ for which $(G,\sigma,\tau)$ is in state \Sing{c}, \Bad{c}, \Good{c} respectively.\label{def:badgood}
\end{defn}

\begin{obs}
Let $(G,\sigma,\tau)$ be any neighboring coloring pair. For $c = \sigma(v),\tau(v)$, if $\delta_c > 0$, then $\sigma,\tau$ are in state \Good{c}.\label{obs:always}
\end{obs}

\begin{proof}
	Let $c\in\{\sigma(v),\tau(v)\}$ and suppose $\delta_c > 0$. If $\delta_c = 1$, then by definition $\sigma,\tau$ are in state \Good{c}. If $\delta_c \ge 2$, then because $A_c = 0$ for $c = \sigma(v),\tau(v)$ by Remark~\ref{remark:specialcase}, $\sigma,\tau$ must be in state \Good{c}.
\end{proof}

\Sing{c} corresponds to the first three constraints of Linear Program~\ref{def:mostbasictight},\footnote{Here, for colors $c\neq\sigma(v),\tau(v)$, we do not distinguish between $c$ for which $\delta_c = 1$ and $(A,B)\in\{(1,2),(1,3),(1,4),(2,1),(3,1),(4,1)\}$ versus colors for which $\delta_c = 1$ and $(A,B)\not\in\{(1,2),(1,3),(1,4),(2,1),(3,1),(4,1)\}$, even though the latter do not correspond to the tight constraints that appear in Linear Program~\ref{def:mostbasictight}. This is just for simplicity of analysis, though making this distinction should yield additional improvements upon our main result. On the other hand, if $c = \sigma(v),\tau(v)$, then $A_c = 0$ by Remark~\ref{remark:specialcase} and we know by Observation~\ref{obs:slack} that this does not correspond to a tight constraint in Linear Program~\ref{def:mostbasictight}.} \Bad{c} corresponds to the last constraint of Linear Program~\ref{def:mostbasictight}, and \Good{c} corresponds to $c$ for which $(A_c,B_c,\vec{a}^c,\vec{b}^c)$ does not correspond to a constraint of Linear Program~\ref{def:mostbasictight}.


\definecolor{vcolor}{RGB}{0,84,147}
\definecolor{ccolor}{RGB}{79,143,0}

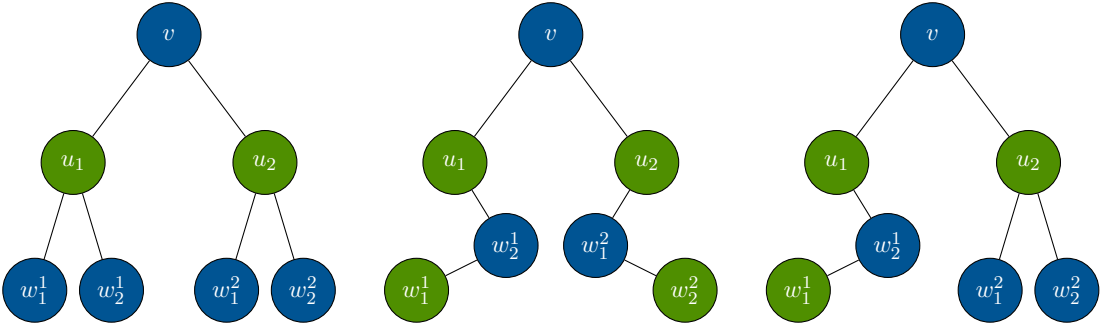
\begin{figure}[ht]\centering
\subcaptionbox{}{
\begin{tikzpicture}[scale=0.85,every node/.style={scale=0.85}]
	\node[shape=circle,draw=black,text=white,fill=vcolor,minimum size=1cm] (v) at (0,0) {$v$};
	\node[shape=circle,draw=black,text=white,fill=ccolor,minimum size=1cm] (u1) at (-1.5,-2) {$u_1$};
	\node[shape=circle,draw=black,text=white,fill=ccolor,minimum size=1cm] (u2) at (1.5,-2) {$u_2$};
	\node[shape=circle,draw=black,text=white,fill=vcolor,minimum size=1cm] (w11) at (-2.1,-4) {$w^1_1$};
	\node[shape=circle,draw=black,text=white,fill=vcolor,minimum size=1cm] (w12) at (-0.9,-4) {$w^1_2$};
	\node[shape=circle,draw=black,text=white,fill=vcolor,minimum size=1cm] (w21) at (0.9,-4) {$w^2_1$};
	\node[shape=circle,draw=black,text=white,fill=vcolor,minimum size=1cm] (w22) at (2.1,-4) {$w^2_2$};
    \path [-] (v) edge node[left] {} (u1);
    \path [-] (v) edge node[left] {} (u2);
    \path [-] (u1) edge node[left] {} (w11);
    \path [-] (u1) edge node[left] {} (w12);
    \path [-] (u2) edge node[left] {} (w21);
    \path [-] (u2) edge node[left] {} (w22);
\end{tikzpicture}
}\quad
\subcaptionbox{}{
\begin{tikzpicture}[scale=0.85,every node/.style={scale=0.85}]
	\node[shape=circle,draw=black,text=white,fill=vcolor,minimum size=1cm] (v) at (0,0) {$v$};
	\node[shape=circle,draw=black,text=white,fill=ccolor,minimum size=1cm] (u1) at (-1.5,-2) {$u_1$};
	\node[shape=circle,draw=black,text=white,fill=ccolor,minimum size=1cm] (u2) at (1.5,-2) {$u_2$};
	\node[shape=circle,draw=black,text=white,fill=ccolor,minimum size=1cm] (w11) at (-2.1,-4) {$w^1_1$};
	\node[shape=circle,draw=black,text=white,fill=vcolor,minimum size=1cm] (w12) at (-0.7,-3.3) {$w^1_2$};
	\node[shape=circle,draw=black,text=white,fill=vcolor,minimum size=1cm] (w21) at (0.7,-3.3) {$w^2_1$};
	\node[shape=circle,draw=black,text=white,fill=ccolor,minimum size=1cm] (w22) at (2.1,-4) {$w^2_2$};
    \path [-] (v) edge node[left] {} (u1);
    \path [-] (v) edge node[left] {} (u2);
    \path [-] (w12) edge node[left] {} (w11);
    \path [-] (u1) edge node[left] {} (w12);
    \path [-] (u2) edge node[left] {} (w21);
    \path [-] (w21) edge node[left] {} (w22);
\end{tikzpicture}
}\quad
\subcaptionbox{}{
\begin{tikzpicture}[scale=0.85,every node/.style={scale=0.85}]
	\node[shape=circle,draw=black,text=white,fill=vcolor,minimum size=1cm] (v) at (0,0) {$v$};
	\node[shape=circle,draw=black,text=white,fill=ccolor,minimum size=1cm] (u1) at (-1.5,-2) {$u_1$};
	\node[shape=circle,draw=black,text=white,fill=ccolor,minimum size=1cm] (u2) at (1.5,-2) {$u_2$};
	\node[shape=circle,draw=black,text=white,fill=ccolor,minimum size=1cm] (w11) at (-2.1,-4) {$w^1_1$};
	\node[shape=circle,draw=black,text=white,fill=vcolor,minimum size=1cm] (w12) at (-0.7,-3.3) {$w^1_2$};
	\node[shape=circle,draw=black,text=white,fill=vcolor,minimum size=1cm] (w21) at (0.9,-4) {$w^2_1$};
	\node[shape=circle,draw=black,text=white,fill=vcolor,minimum size=1cm] (w22) at (2.1,-4) {$w^2_2$};
    \path [-] (v) edge node[left] {} (u1);
    \path [-] (v) edge node[left] {} (u2);
    \path [-] (w12) edge node[left] {} (w11);
    \path [-] (u1) edge node[left] {} (w12);
    \path [-] (u2) edge node[left] {} (w21);
    \path [-] (u2) edge node[left] {} (w22);
\end{tikzpicture}
}
\caption{$(G,\sigma,\tau)$ in state \Bad{\text{green}} \---- only $\sigma$ shown}
\label{fig:treeexamples}
\end{figure}

$N_{sing}(\sigma,\tau)$ can be large even for a ``typical'' neighboring coloring: consider any $(G,\sigma,\tau)$ where the neighbors of $v$ are all connected to each other. Indeed, this example is the reason that all existing results on sampling colorings that followed \cite{vigoda2000improved} needed to at least assume triangle-freeness of $G$, otherwise the uniformity properties they leverage simply do not hold.

Instead of avoiding state \Sing{c}, we want ``typical'' neighboring coloring pairs to avoid \Bad{c} for many $c$ (of course, we still need to make precise what we mean by ``typical,'' which we defer to Section~\ref{sec:coupling}).

Consider the following thought experiment. Let $\mathcal{C}$ be the family of all neighboring coloring pairs such that for every $(G,\sigma,\tau)\in\mathcal{C}$, \begin{equation}N_{bad}(\sigma,\tau)\le\gamma\cdot N_{good}(\sigma,\tau)\label{eq:comparable}\end{equation} for some absolute constant $\gamma > 0$. Now suppose $k = (11/6-\epsilon)d$ for some small absolute constant $\epsilon > 0$, and our goal was just to design a greedy coupling so that every neighboring coloring pair in some family contracts. Observation~\ref{obs:slack} and complementary slackness intuitively suggest that this should be possible for a small enough $\epsilon$ depending only on $\gamma$.

To get an effective estimate for $\epsilon$ in this thought experiment, we need to encode \eqref{eq:comparable} into Linear Program~\ref{def:lp}. This motivates the following modified linear program, which captures a type of amortized analysis that we will describe in more detail shortly. 

\begin{lp}
Fix some $N_{max}\ge 1$, $m^*\ge 2$, and $\gamma > 0$. For variables $\{p_{\alpha}\}_{\alpha\in\N}$, $\lambda_{sing}$, $\lambda_{bad}$, $\lambda_{good}$, and $\lambda$, and dummy variables $x, y$, minimize $\lambda$ subject to the following constraints: $0\le p_{\alpha}\le p_{\alpha-1}\le p_1 = 1$ for all $\alpha\ge 2$, constraint \eqref{eq:NMax}, and constraint \eqref{eq:pap}. For all $(A,B,\vec{a},\vec{b})$ for which $m < m^*$, $\vec{a},\vec{b}\in\{0,1,\cdots,N_{max}\}^m\backslash\{(0,0,\cdots 0)\}$, and $A,B$ satisfy \eqref{eq:trivABbound}, define a constraint $$H(A,B,\vec{a},\vec{b}) \le -1 + \lambda_s\cdot m,$$ where $\lambda_s$ is:
\begin{itemize}[noitemsep]
	\item $\lambda_{sing}$ if $m = 1$.
	\item $\lambda_{bad}$ if $m = 2$ and either $(a_1,a_2,b_1,b_2) = (1,1,3,3)$ and $B=7$, or $(a_1,a_2,b_1,b_2) = (3,3,1,1)$ and $A=7$.
	\item $\lambda_{good}$ otherwise.
\end{itemize} Define a constraint $$(B - b_m)p_B + \sum_i b_i p_{b_i} \le -1 + \lambda_{good}\cdot m$$ for all $2\le m < m^*$, $0\le b_1\le\cdots\le b_m\le N_{max}$ where $b_m > 0$, and $B = \sum_i b_i$. Define the constraints $$x\ge (A - 2)p_A$$ $$y\ge a\cdot p_a + b\cdot p_b - \min(p_a,p_b)$$ $$-1 + \lambda_{good}\cdot m^* \ge 2x + m^* y$$ for every $A,a,b$ satisfying $0\le A\le 1+N_{max}$ and $0\le a < b\le N_{max}$. Finally, define the constraints $$\lambda \ge \lambda_{sing}, \ \ \ \lambda\ge\lambda_{good}$$ $$\lambda\ge\frac{\gamma}{\gamma+1}\cdot\lambda_{bad} + \frac{1}{\gamma+1}\cdot\lambda_{good}.$$ Call this the \emph{$\gamma$-mixed coupling LP} and denote its objective value by $\lambda^*_{\gamma}$.
\label{def:mixedlp}
\end{lp}

Denote a minimizing choice of $\{p_{\alpha}\}$ and $\lambda_{sing},\lambda_{bad},\lambda_{good}$ for the $\gamma$-mixed coupling LP by $\{p^*_{\alpha}\}$ and $\lambda^*_{sing},\lambda^*_{bad},\lambda^*_{good}$. As discussed, Observation~\ref{obs:slack} and complementary slackness suggest that $\lambda^*_{\gamma} < 11/6$, and for any fixed $\gamma$ we can get an effective estimate on $\lambda^*_{\gamma}$ simply by solving the LP.

Then for any $(G,\sigma,\tau)\in\mathcal{C}$ and the greedy coupling $(\sigma,\tau)\mapsto(\sigma',\tau')$, observe that \begin{align}\E[d(\sigma',\tau') - 1] &\le -|\{c: \delta_c = 1\}| + \sum_{\substack{c: \sigma,\tau \\ \text{\Sing{c}}}} H(A_c,B_c,\vec{a}^c,\vec{b}^c) + \sum_{\substack{c: \sigma,\tau \\ \text{\Bad{c}}}} H(A_c,B_c,\vec{a}^c,\vec{b}^c) + \sum_{\substack{c: \sigma,\tau \\ \text{\Good{c}}}} H(A_c,B_c,\vec{a}^c,\vec{b}^c)\nonumber\\
&\le -|\{c: \delta_c = 1\}| + \sum_{\substack{c: \sigma,\tau \\ \text{\Sing{c}}}}(-1 + \lambda_{sing}) + \sum_{\substack{c: \sigma,\tau \\ \text{\Bad{c}}}} (-1 + 2\lambda_{bad}) + \sum_{\substack{c: \sigma,\tau \\ \text{\Good{c}}}} (-1 + \delta_c\lambda_{good})\nonumber\\
&= -k + \lambda_{sing}\cdot N_{sing}(\sigma,\tau) + 2\lambda_{bad}\cdot N_{bad}(\sigma,\tau) + \sum_{\substack{c: \sigma,\tau \\ \text{\Good{c}}}}\delta_c\cdot \lambda_{good}.\label{eq:basicconvexcombo}
\end{align} But because $\delta_c\ge 2$ for any $c\neq\sigma(v),\tau(v)$ for which $\sigma,\tau$ are in state \Good{c}, because $\sigma,\tau$ are always in state \Good{\sigma(v)}, \Good{\tau(v)}, and because $$N_{sing}(\sigma,\tau) + 2N_{bad}(\sigma,\tau) + \sum_{c:\sigma,\tau \ \text{\Good{c}}} \delta_c = d(v),$$ we conclude that \eqref{eq:basicconvexcombo} is a convex combination of the terms $$-k + \lambda^*_{sing}\cdot d(v), \quad -k + \lambda^*_{good}\cdot d(v), \quad-k + \left(\frac{\gamma}{\gamma+1}\cdot \lambda^*_{bad} + \frac{1}{\gamma + 1}\cdot \lambda^*_{good}\right)d(v).$$ So we conclude that $\E[d(\sigma',\tau') - 1]\le -k + \lambda^*_{\gamma}d(v) < 0$ as long as $k> \lambda^*_{\gamma}d$.

In the next section, we go from the intuition of this thought experiment to a rigorous notion of ``typical'' neighboring coloring pairs avoiding the state \Bad{c}. We already reduced finding a coupling for all of $\mathcal{C}$ to analyzing the $\gamma$-mixed coupling LP, and in the sequel we will reduce finding a coupling for \emph{all} neighboring coloring pairs to analyzing the $\gamma$-mixed coupling LP.


\section{Avoiding Worst-Case Neighborhoods in Expectation}
\label{sec:coupling}

The key idea is that regardless of what neighboring coloring pair $(G,\sigma,\tau)$ one starts with, the probability that $\sigma',\tau'$ derived from one step of greedy coupling has changed in distance is $\Theta(1/n)$ (see Lemma~\ref{lem:probend} below). So in expectation, one can run $\Theta(n)$ steps of greedy coupling before the two colorings either coalesce to the same coloring or have Hamming distance greater than 1, but by that time the set of colors around $v$ will have changed substantially. This is the main insight of \cite{dyer2001extension,hayes2007variable}, who leverage it to analyze the Glauber dynamics and slightly improve upon Jerrum's $k\ge 2d$ bound under extra girth and degree assumptions.

We leverage this insight to analyze the flip dynamics under no extra assumptions. Recall that at the end of Section~\ref{sec:tightexamples} we showed that if $k > \lambda^*_{\gamma}d$, then there exist flip probabilities for which one step of greedy coupling will contract any $(G,\sigma,\tau)$ for which $N_{bad}(\sigma,\tau)\le\gamma\cdot N_{good}(\sigma,\tau)$. In this section, after formalizing the variable-length coupling alluded to above, we show that if $k > \lambda^*_{O(\gamma)}\cdot d$, then there exist flip probabilities for which the variable-length coupling will contract any neighboring coloring pair because by the end of the coupling, it will satisfy $N_{bad}(\sigma,\tau)\le \gamma\cdot N_{good}(\sigma,\tau)$ in expectation.

\subsection{The Variable-Length Coupling}

Our variable-length coupling simply runs the greedy coupling until the distance between the two colorings changes.

\begin{enumerate}
	\item Start with two neighboring colorings $\sigma^{(0)},\tau^{(0)}$.
	\item Initialize $t = 1$. Repeat:
	\begin{enumerate}
	 	\item Run the greedy one-step coupling of Section~\ref{subsec:introonestep} to flip components $S_t$ in $\sigma^{(t-1)}$ and $S'_t$ in $\tau^{(t-1)}$, producing $\sigma^{(t)},\tau^{(t)}$ (note that $S_t$ or $S'_t$ might be empty, e.g with probability $1 - p_{\alpha}$, a component of size $\alpha$ that is chosen to be flipped is not actually flipped).
	 	\item If $d(\sigma^{(t)},\tau^{(t)})\neq d(\sigma^{(t-1)},\tau^{(t-1)})$, then terminate and define $\Tstop = t$. Otherwise, increment $t$.
	 \end{enumerate}

	 We call any subsequence of pairs of flips $(S_i,S'_i),...,(S_j,S'_j)$ a \emph{coupling schedule} starting from the neighboring coloring pair $(G,\sigma^{(i-1)},\tau^{(i-1)})$.
\end{enumerate}

It is easy to see that this satisfies the conditions of being a variable-length coupling as in Definition~\ref{def:varlength}. Indeed it is the same coupling as in \cite{hayes2007variable}, except it is for the flip dynamics instead of the Glauber dyanmics. Note that we have a characterization of the pairs of flips $(S_t,S'_t)$ which terminate the coupling: at least one of them must belong to the symmetric difference $D$ defined in Section~\ref{subsec:introonestep}.

\begin{defn}
	Given a neighboring coloring pair $(G,\sigma,\tau)$, a pair of components $S$ in $\sigma$ and $S'$ in $\tau$ is \emph{terminating} if $S = S_{\sigma}(v,c)$ or $S' = S_{\tau}(v,c)$, or there exists $u\in N(v)$ for which $S = S_{\sigma}(u,\tau(v))$ or $S' = S_{\tau}(u,\sigma(v))$.
\end{defn}

Note that for any $t$, $S_t$ and/or $S'_t$ may be the empty set. Moreover, because flips of components outside of $D$ are matched via the identity coupling, if $(S_t,S'_t)$ is not terminating, then $S_t = S'_t$.

\begin{lem}
Let components $S$ in $\sigma$ and $S'$ in $\tau$ be chosen according to the greedy coupling. Then $$\frac{k - d - 2}{nk}\le \P[(S,S') \ \text{terminating}] \le \frac{k + 2p_2d}{nk}.$$\label{lem:probend}
\end{lem}

\begin{proof}
	For the lower bound, note that the pair $(S_{\sigma}(v,c),S_{\tau}(v,c))$ is terminating for any $c$. In particular, for $c\neq\{\sigma(v),\tau(v)\}$ such that $\delta_c = 0$, $S_{\sigma}(v,c) = S_{\tau}(v,c) = \{v\}$ \---- note that while the vertex sets for these components are all $\{v\}$, the flips are all distinct as they vary with $c$. Each such pair of flips has probability mass $(1/nk)\cdot p_1 = (1/nk)$, and there are at least $k - d - 2$ such choices of $c$, giving the lower bound.

	For the upper bound, fix a color $c$ for which $\delta_c\neq 0$ and some $i\in[\delta_c]$. For $i\neq i_{max},j_{max}$, the pairs of flips $(S_{\sigma}(u_i,\tau(v)),S_{\tau}(u_i,\sigma(v)))$, $(S_{\sigma}(u_i,\tau(v)),\emptyset)$, and $(\emptyset,S_{\tau}(u_i,\sigma(v)))$ have probability mass $\min(p_{b_i},p_{a_i}), \max(0,p_{b_i}-p_{a_i})$, and $\max(p_{a_i} - p_{b_i},0)$, for a total of $\max(p_{a_i},p_{b_i})$. The remaining pairs of flips have probabily masses which depend on whether $i_{max} = j_{max}$, as shown in Table~\ref{table:masses}.


\begin{table}[ht]
\centering
\caption{Probability masses for some coupled flips}
\label{table:masses}
\begin{tabularx}{\linewidth}{|c|c|Y|Y|}\hline
Flip in $\sigma$                   & Flip in $\tau$                    & $i_{max} = j_{max}$                                     & $i_{max}\neq j_{max}$                                                                       \\ \hline
$S_{\sigma}(v,c)$                  & $S_{\tau}(u_{i_{max}},\sigma(v))$ & $p_A$                                                   & $p_A$                                                                                       \\ \hline
$S_{\tau}(v,c)$                    & $S_{\sigma}(u_{j_{max}},\tau(v))$ & $p_B$                                                   & $p_B$                                                                                       \\ \hline
$S_{\sigma}(u_{i_{max}},\tau(v))$  & $S_{\tau}(u_{i_{max}},\sigma(v))$ & $\min(p_{a_{i_{max}}}-p_A,p_{b_{i_{max}}}-p_B)$         & $\min(p_{a_{i_{max}}}-p_A,p_{b_{i_{max}}})$                                                 \\ \hline
$(S_{\sigma}(u_{j_{max}},\tau(v))$ & $S_{\tau}(u_{j_{max}},\sigma(v))$ & N/A                                                  & $\min(p_{b_{j_{max}}}-p_B,p_{a_{j_{max}}})$                                                 \\ \hline
$(S_{\sigma}(u_{i_{max}},\tau(v))$ & $\emptyset$                       & $\max(0,p_{b_{i_{max}}} - p_B - p_{a_{i_{max}}} + p_A)$ & $\max(0,p_{b_{i_{max}}} - p_{a_{i_{max}}} + p_A)$                                           \\ \hline
$\emptyset$                        & $S_{\tau}(u_{i_{max}},\sigma(v))$ & $\max(0,p_{a_{i_{max}}} - p_A - p_{b_{i_{max}}} + p_B)$ & $\max(0,p_{a_{i_{max}}} - p_A - p_{b_{i_{max}}})$                                           \\ \hline
$S_{\sigma}(u_{j_{max}},\tau(v))$  & $\emptyset$                       & N/A                                                  & $\max(0,p_{a_{j_{max}}} - p_{b_{j_{max}}} + p_B)$                                           \\ \hline
$\emptyset$                        & $S_{\tau}(u_{j_{max}},\sigma(v))$ & N/A                                                  & $\max(0,p_{b_{j_{max}}} - p_B - p_{a_{j_{max}}})$                                           \\ \hline\hline
\multicolumn{2}{|c|}{Total}                             & $\max(p_{a_{i_{max}}} + p_A,p_{b_{i_{max}}} + p_B)$     & $\max(p_{a_{j_{max}}} + p_B,p_{b_{j_{max}}}) + \max(p_{b_{i_{max}}} + p_A,p_{a_{i_{max}}})$\\ \hline
\end{tabularx}
\end{table}

	From these we can conclude that $$nk\cdot\Pr[(S,S') \ \text{terminating}] \le \left(\sum_{c:\delta_c > 0}p_{A_c} + p_{B_c}\right) + \left(\sum_{c:\delta_c > 0,i\in[\delta_c]}\max(p_{a^c_i},p_{b^c_i})\right) + \left(\sum_{c: \delta_c = 0}p_1\right).$$ The sum of the second and third summands is at most $k$. For the first summand, note that when $A_c,B_c$ are nonzero and $\delta_c > 0$, $A_c,B_c\ge 2$, so the first summand is at most $d\cdot(2p_2)$. The desired upper bound follows.
\end{proof}

\begin{cor}
	$\max_{\sigma^{(0)},\tau^{(0)}}\E[\Tstop]\le\frac{nk}{k - d - 2}$.\label{eq:maxET}
\end{cor}

\subsection{Reduction to Comparing \texorpdfstring{$N_{bad}(\sigma^{(\Tstop-1)},\tau^{(\Tstop-1)})$ and $N_{good}(\sigma^{(\Tstop-1)},\tau^{(\Tstop-1)})$}{Nbad(Tstop) vs. Ngood(Tstop)}}

We now give a reduction from analyzing the expected change in distance under our variable-length coupling to proving that the relation \eqref{eq:comparable} from our thought experiment holds in expectation by the end of the coupling.

\begin{lem}
	Suppose there exists a constant $\gamma > 0$ for which $$\E[N_{bad}(\sigma^{(\Tstop - 1)},\tau^{(\Tstop - 1)})]\le\gamma\cdot\E[N_{good}(\sigma^{(\Tstop - 1)},\tau^{(\Tstop - 1)})]$$ for any initial neighboring coloring pair $(G,\sigma^{(0)},\tau^{(0)})$.

	Let $\lambda^*_{C\gamma}$ be the objective value of the $C\gamma$-mixed coupling LP, where $C := \frac{k + 2p_2d}{k - d - 2}$. Then $$\E[d(\sigma^{(\Tstop)},\tau^{(\Tstop)})-1]\le \frac{-k+\lambda^*_{C\gamma}d}{k - d - 2}.$$\label{lem:reduction}
\end{lem}

\begin{proof}
This mainly just follows by linearity of expectation and the calculation done at the end of Section~\ref{subsec:modifyinglp}, with the slight complication that the probability that the coupling terminates at any given point is only the same up to constant factors.

Let $\lambda^*_{sing}, \lambda^*_{tree}, \lambda^*_{good}$ be the values for $\lambda_{sing},\lambda_{tree},\lambda_{good}$ of the minimizer of the $C\gamma$-mixed coupling LP from Definition~\ref{def:mixedlp}. For alternating components $S,S'$ in $\sigma,\tau$ respectively, define $E^{S,S'}_{\sigma,\tau} = d(\sigma',\tau') - 1$, where $\sigma',\tau'$ is the pair of colorings obtained by flipping $S$ in $\sigma$ and $S'$ in $\tau$, and let $p^{S,S'}_{\sigma^{(T-1)},\tau^{(T-1)}}$ be the probability that $S,S'$ are flipped in one step of greedy coupling starting from $\sigma^{(T-1)}, \tau^{(T-1)}$. Then we have that \begin{equation}
	\E[d(\sigma^{(\Tstop)},\tau^{(\Tstop)})-1] = \sum_{T,\sigma,\tau}\Pr[\sigma^{(T-1)} = \sigma,\tau^{(T-1)}=\tau]\cdot Z(\sigma^{(T-1)},\tau^{(T-1)}),\label{eq:expdiff}
\end{equation} where $$Z(\sigma^{(T-1)},\tau^{(T-1)}) \triangleq \sum_{S,S'}\mathbb{I}[(S,S') \ \text{terminating}]\cdot p^{S,S'}_{\sigma^{(T-1)},\tau^{(T-1)}}\cdot E^{S,S'}_{\sigma^{(T-1)},\tau^{(T-1)}}$$ But note that for $v^{(T)},c^{(T)}$ not terminating, $E^{S,S'}_{\sigma^{(T-1)},\tau^{(T-1)}} = 0$ because the one-step coupling is just the identity coupling, so $Z(\sigma^{(T-1)},\tau^{(T-1)})$ is just the expected change in distance under one step of greedy coupling on the neighboring coloring pair $(G,\sigma^{(T-1)},\tau^{(T-1)})$. Therefore, for any $T\le\Tstop$, \begin{align}
	Z(\sigma^{(T-1)},\tau^{(T-1)}) &= \E[d(\sigma^{(T)},\tau^{(T)}) - 1] \nonumber\\
	&\le \frac{1}{nk}\Big((-1 + \lambda^*_{sing})\cdot N_{sing}(\sigma^{(T-1)},\tau^{(T-1)}) + (-1 + 2\lambda^*_{bad})\cdot N_{bad}(\sigma^{(T-1)},\tau^{(T-1)}) \nonumber\\
	&\qquad \quad + \sum_{\substack{c: \sigma^{(T-1)},\tau^{(T-1)} \\ \text{\Good{c}}}} (-1 + \delta_c\cdot \lambda^*_{good}) - |\{c: \delta_c = 0\}|\Big) \nonumber\\
	&= \frac{1}{nk}\cdot\Big(-k + \lambda^*_{sing}N_{sing}(\sigma^{(T-1)},\tau^{(T-1)}) + 2\lambda^*_{bad}N_{bad}(\sigma^{(T-1)},\tau^{(T-1)}) \nonumber\\
	&\qquad \quad + \sum_{\substack{c: \sigma^{(T-1)},\tau^{(T-1)} \\ \text{\Good{c}}}}\delta_c\cdot \lambda^*_{good}\Big).\label{eq:Z}
\end{align} Because $$N_{sing}(\sigma,\tau) + 2N_{bad}(\sigma,\tau) +\sum_{\substack{c: \sigma,\tau \\ \text{\Good{c}}}}\delta_c = d(v)$$ for all neighboring coloring pairs $(G,\sigma,\tau)$, we conclude from \eqref{eq:expdiff} and \eqref{eq:Z} that \begin{equation}\E[d(\sigma^{(\Tstop)},\tau^{(\Tstop)}) - 1]=\frac{1}{nk}\left(\sum_{T,\sigma,\tau}\Pr[\sigma^{(T-1)} = \sigma,\tau^{(T-1)} = \tau]\right)\left(-k + \lambda^*_{sing}d_{sing} + \lambda^*_{bad}d_{bad} + \lambda^*_{good}d_{good}\right)\label{eq:convexcombo}\end{equation} for some $d_{sing},d_{bad},d_{good}\ge 0$ satisfying \begin{equation}d_{sing} + d_{bad} + d_{good} = d(v).\label{eq:addtodv}\end{equation} But note that \begin{align}
	\frac{d_{bad}}{d_{good}} &= \frac{\sum_{T,\sigma,\tau}\Pr[\sigma^{(T-1)} = \sigma,\tau^{(T-1)} = \tau]\cdot N_{bad}(\sigma,\tau)}{\sum_{T,\sigma,\tau}\Pr[\sigma^{(T-1)} = \sigma,\tau^{(T-1)} = \tau]\cdot N_{good}(\sigma,\tau)} \nonumber\\
	&\le C\cdot\frac{\E[N_{bad}(\sigma^{(\Tstop - 1)},\tau^{(\Tstop - 1)})]}{\E[N_{good}(\sigma^{(\Tstop - 1)},\tau^{(\Tstop - 1)})]} \nonumber\\
	&\le C\gamma\label{eq:dbadtodgood}
\end{align} for $C := \frac{k + 2p_2d}{k - d - 2}$, where the second inequality follows by hypothesis and the first inequality follows by the fact that for $s\in\{\text{bad},\text{good}\}$, \begin{align*}
	\E[N_{s}(\sigma^{(\Tstop - 1)},\tau^{(\Tstop - 1)})] &= \sum_{T,\sigma,\tau}\Pr[\sigma^{(T-1)} = \sigma,\tau^{(T-1)} = \tau]\cdot \Pr[(S_T,S'_T) \ \text{terminating}]\cdot N_{s}(\sigma,\tau) \\
	&\in \left[\frac{k - d - 2}{nk},\frac{k + 2p_2d}{nk}\right]\cdot \sum_{T,\sigma,\tau}\Pr[\sigma^{(T-1)} = \sigma,\tau^{(T-1)} = \tau]\cdot N_{s}(\sigma,\tau),
\end{align*} where we use the notation $x \in [a,b]\cdot y$ to denote the fact that $a\cdot y\le x\le b\cdot y$. The first step above follows by definition and the second step follows by Lemma~\ref{lem:probend}.

Finally, observe that \eqref{eq:addtodv} and \eqref{eq:dbadtodgood} imply that $-k + \lambda^*_{sing}d_{sing} + \lambda^*_{bad}d_{bad} + \lambda^*_{good}d_{good}$ is a convex combination of $-k + \lambda^*_{sing}d(v)$, $-k + \lambda^*_{good}d(v)$, and $-k + \left(\frac{C\gamma}{C\gamma + 1}\lambda^*_{bad} + \frac{1}{C\gamma + 1}\lambda^*_{good}\right)d(v)$, so in particular from \eqref{eq:convexcombo} we get that $$\E[d(\sigma^{(\Tstop)},\tau^{(\Tstop)}) - 1]\le\frac{1}{nk}\left(\sum_{T,\sigma,\tau}\Pr[\sigma^{(T-1)} = \sigma,\tau^{(T-1)} = \tau]\right)(-k + \lambda^*_{C\gamma}d(v))\le \frac{-k + \lambda^*_{C\gamma}d(v)}{k - d - 2},$$ where the final step follows from the fact that $$\sum_{T,\sigma,\tau}\Pr[\sigma^{(T-1)} = \sigma,\tau^{(T-1)} = \tau]\cdot\Pr[(S_T,S'_T) \ \text{terminating}] = 1$$ and the lower bound of Lemma~\ref{lem:probend}.
\end{proof}


\section{Burn-in for the Flip Dynamics}
\label{sec:rand_process}

In this section we prove our main technical lemma. Throughout, assume that $k\ge 1.833d$.

\begin{lem}
Suppose flip probabilities $\{p_{\alpha}\}_{\alpha\in\N}$ satisfy $0\le p_{\alpha}\le p_{\alpha-1}\le p_1 = 1$ for all $\alpha\ge 2$, constraint \eqref{eq:pap}, and additionally $\alpha p_{\alpha - 2}\le 3$ for all $\alpha\ge 3$.

For $$\gamma\triangleq \frac{(6k - d - 2)(k + 2p_2d)}{4(k - d - 2)(k - d - 1)},$$ we have that $$\E[N_{bad}(\sigma^{(\Tstop - 1)},\tau^{(\Tstop - 1)})]\le\gamma\cdot\E[N_{good}(\sigma^{(\Tstop - 1)},\tau^{(\Tstop - 1)})]$$ for any initial neighboring coloring pair $(G,\sigma^{(0)},\tau^{(0)})$.\label{lem:main}
\end{lem}

\begin{remark}
The additional constraint that $\alpha p_{\alpha - 2}\le 3$ in fact already holds for the solutions to the $\gamma$-mixed coupling LP that we will consider, so we assume it just to obtain better constant factors in our analysis below.
\end{remark}

We first make a simple reduction. Fix any initial neighboring coloring pair $(G,\sigma^{(0)},\tau^{(0)})$, and for every color $c$ denote by $p_{bad}(c)$ and $p_{good}(c)$ the probability that $\sigma^{(\Tstop)},\tau^{(\Tstop)}$ is in state \Bad{c} and \Good{c}, respectively.

By linearity of expectation we have that $$\E[N_{bad}(\sigma^{(\Tstop)},\tau^{(\Tstop)})] = \sum_{c}p_{bad}(c), \ \ \ \ \E[N_{good}(\sigma^{(\Tstop)},\tau^{(\Tstop)})] = \sum_{c}p_{good}(c).$$ Therefore to show Lemma~\ref{lem:main}, it is enough to show the following.

\begin{lem}
 	$p_{bad}(c)\le\gamma\cdot p_{good}(c)$ for every color $c$.\label{lem:realmain}
\end{lem} 

This is certainly true for $c = \sigma^{(0)}(v),\tau^{(0)}(v)$, in which case $p_{bad}(c) = 0$ and $p_{good}(c) = 1$ by Observation~\ref{obs:always}. We point out that while the case of $c = \sigma^{(0)}(v),\tau^{(0)}(v)$ is the one for which the fact that our state space includes all colorings, improper and proper, introduces complications in the definition of the greedy coupling (see Remark~\ref{remark:specialcase}), it happens to be the easiest case of Lemma~\ref{lem:realmain}.

So henceforth assume $c\neq\sigma^{(0)}(v),\tau^{(0)}(v)$. We proceed via a fractional matching argument. Take any coupling schedule $\Sigma_{pre} = (S_1,S_1), (S_2,S_2), \cdots, (S_{T-1},S_{T-1})$ consisting of pairs of identical flips, and define $\mathcal{W}$ to be the set of all coupling schedules of the form $(S_1,S_1), (S_2,S_2), \cdots, (S_{T-1},S_{T-1}), (S_T, S'_T)$ for $(S_T,S'_T)$ terminating. In other words, $\mathcal{W}$ consists of all $T$-step coupling schedules whose first $T - 1$ steps are fixed to $\Sigma_{pre}$ and which only changes the distance between the colorings in the last step $(S_T,S'_T)$. We will match to the collection of schedules $\mathcal{W}$ an (infinite) collection of schedules of the following form.

\begin{defn}
	Fix $S_1,...,S_{T-1}$. A coupling schedule $\Sigma_*$ starting from the neighboring coloring pair $(G,\sigma^{(0)},\tau^{(0)})$ is \emph{satisfying} if it is of the form \begin{equation}\Sigma_* = (S_1,S_1),\cdots,(S_{T-1},S_{T-1}),(S^*_T,S^*_T),\cdots, (S^*_{T^*-1}, S^*_{T^*-1}), (S^*_{T^*},S'^*_{T^*})\label{eq:defsatisfying}\end{equation} for $(S^*_{T^*},S'^*_{T^*})$ terminating, and gives rise to a sequence of colorings $$(\sigma^{(0)},\tau^{(0)}), (\sigma^{(1)},\tau^{(1)}),\cdots,(\sigma^{(T-1)},\tau^{(T-1)}),(\sigma^{(T)}_*,\tau^{(T)}_*),...,(\sigma^{(T^*)}_*,\tau^{(T^*)}_*)$$ for which \begin{enumerate}
 	\item $\sigma^{(T^*-1)}_*,\tau^{(T^*-1)}_*$ are in state \Good{c}
 	\item $\sigma^{(t)}_*,\tau^{(t)}_*$ is not in state \Bad{c} for any $T\le t<T^*$.
	\end{enumerate}\label{def:satisfying}
\end{defn}

Property 2) in Definition~\ref{def:satisfying} ensures that from any satisfying schedule $\Sigma_*$, we can uniquely decode the collection $\mathcal{W}$ to which it is being fractionally matched: in $\Sigma_*$, take the last pair of colorings in state \Bad{c}, and $\Sigma_{pre}$ is the subsequence of $\Sigma_*$ starting from $(\sigma^{(1)},\tau^{(1)})$ and ending at that pair.

If we can show that for any $\Sigma_{pre} = S_1,...,S_{T-1}$ $$\sum_{\Sigma_* \ \text{satisfying}}\Pr[(S^*_T,S^*_T),\cdots, (S^*_{T^*-1}, S^*_{T^*-1}), (S^*_{T^*},S'^*_{T^*})\vert \Sigma_{pre}]\ge \frac{1}{\gamma}\cdot\frac{k + 2p_2d}{nk},$$ then because in general, $$\Pr[(S_T,S'_T) \ \text{terminating}\vert\Sigma_{pre}]\le\frac{k+2p_2d}{nk}$$ by the upper bound of Lemma~\ref{lem:probend}, this will imply that $p_{bad}(c)\le\gamma\cdot p_{good}(c)$.

To exhibit such a collection of satisfying coupling schedules $\Sigma_*$, we first define a coarsening of the state space as follows. Starting from the neighboring coloring pair $(G,\sigma^{(T-1)},\tau^{(T-1)})$ which is in state \Bad{c}, take any subsequent coupling schedule \begin{equation}(S^*_T,S^*_T),...,(S^*_{T^*},S'^*_{T^*})\label{extraschedule}\end{equation} with $(S^*_{T^*},S'^*_{T^*})$ terminating which gives rise to a sequence of pairs of colorings \begin{equation}(\sigma^{(T)}_*,\tau^{(T)}_*),...,(\sigma^{(T^*)}_*,\tau^{(T^*)}_*),\label{eq:extra}\end{equation} where $(G,\sigma^{(t)}_*,\tau^{(t)}_*)$ is a neighboring coloring pair for all $t$ except $t = T'$. Define the following auxiliary states. To avoid confusion with the states defined in Definition~\ref{def:badgood}, we will refer to the auxiliary states defined below as \emph{stages}.

\begin{defn}
	Let $c$ be any color, not necessarily one appearing in the neighborhood of $v$. We say that $\sigma^{(t)}_*,\tau^{(t)}_*$ is in \emph{stage \Goodend{c}} if $\sigma^{(t-1)}_*,\tau^{(t-1)}_*$ is in state \Good{c} and the pair of flips $(S,S')$ giving rise to $\sigma^{(t)}_*,\sigma^{(t)}_*$ from $\sigma^{(t-1)}_*,\tau^{(t-1)}_*$ is terminating.

	We say $\sigma^{(t)}_*,\tau^{(t)}_*$ is in \emph{stage \Badend{c}} if, intuitively, we choose to quit looking for satisfying coupling schedules among those of which $(\sigma^{(0)}_*,\tau^{(0)}_*),...,(\sigma^{(t)}_*,\tau^{(t)}_*)$ is a prefix. Formally, $\sigma^{(t)}_*,\tau^{(t)}_*$ is in stage \Badend{c} if least one of the following conditions holds (note that these conditions aren't necessarily mutually exclusive):
	\begin{enumerate}
		\item[(i)] $t = T $ and the pair of flips $(S,S')$ giving rise to $\sigma^{(T)}_*,\sigma^{(T)}_*$ from the initial pair $\sigma^{(T-1)},\tau^{(T-1)}$ is terminating (i.e. if $(S_1,S_1),...,(S_{T-1},S_{T-1}),(S,S')\in\mathcal{W}$).
		\item[(ii)] $t = T$ and $\sigma^{(t)}_*,\tau^{(t)}_*$ is not in state \Good{c}.
		\item[(iii)] $\sigma^{(t-1)}_*,\tau^{(t-1)}_*$ is in state \Good{c} but $\sigma^{(t)}_*,\tau^{(t)}_*$ is not in state \Good{c} or stage \Goodend{c} (this includes the case that $c$ does not appear in the neighborhood of $v$ in $\sigma^{(t)}_*,\tau^{(t)}_*$).
		\item[(iv)] $t > T$ and $\sigma^{(t-1)}_*,\tau^{(t-1)}_*$ is in stage \Badend{c}.
	\end{enumerate}

	If $\sigma^{(t)}_*,\tau^{(t)}_*$ is not in stage \Badend{c} or \Goodend{c} and is in state \Bad{c} (resp. \Good{c}), then we say it is also in \emph{stage \Bad{c}} (resp. \emph{stage \Good{c}}).\label{def:stages}
\end{defn}

Note that if a sequence of the form \eqref{eq:extra} contains a pair of colorings in stage \Goodend{c}, that pair must be $\sigma^{(T^*)}_*,\tau^{(T^*)}_*$. Furthermore, given any sequence \eqref{eq:extra} for which $\sigma^{(T^*)}_*,\tau^{(T^*)}_*$ is in stage \Goodend{c} with associated coupling schedule \eqref{extraschedule}, note that the corresponding coupling schedule $\Sigma_*$ defined in \eqref{eq:defsatisfying} is satisfying, by definition of stage \Badend{c}.

So it is enough to show that if we start from a neighboring coloring pair $(G,\sigma^{(T-1)},\tau^{(T-1)})$ which is in state \Bad{c} and evolve a sequence of pairs of colorings \eqref{eq:extra} according to the greedy coupling at each step, then \begin{equation}\Pr[\sigma^{(T^*)}_*,\tau^{(T^*)}_*) \ \text{are in stage \Goodend{c}}\vert \sigma^{(T-1)},\tau^{(T-1)}]\ge\frac{1}{\gamma}\cdot\frac{k + 2p_2d}{nk}.\label{eq:mainmarkov}\end{equation} It remains to bound the probabilities of the transitions between the different stages of Definition~\ref{def:stages} under the flip dynamics and the greedy coupling (see Figure~\ref{fig:transitions} for a depiction of the transitions that can occur). A key point is that these bounds will be independent of the specific colorings or structure of $G$.


\begin{figure}[ht]
\centering
\begin{tikzpicture}[align=center]
\usetikzlibrary{positioning,shapes.geometric,arrows}
\tikzstyle{box} = [rectangle, rounded corners, minimum width=2cm, minimum height=1cm,text centered, draw=black, fill=white]
\tikzstyle{arrow} = [thin,->,>=stealth]
\tikzstyle{line} = [draw, -latex']

	\node (bad) [box] {\Bad{c}};
	\node (badend) [box, below=0.75cm of bad] {\Badend{c}};
	\node (goodend) [box,right=3cm of badend] {\Goodend{c}};
	\node (good) [box,right=3.1cm of bad] {\Good{c}};
	\draw [arrow] (bad) -- (badend);
	\draw [arrow] (bad) -- (good);
	\draw [arrow] (good) -- (badend);
	\draw [arrow] (good) -- (goodend);
	\path [line] (good) edge[loop right]();
	\path [line] (badend) edge[loop left]();
\end{tikzpicture}
\caption{Possible transitions among stages of Definition~\ref{def:stages}}
\label{fig:transitions}
\end{figure}
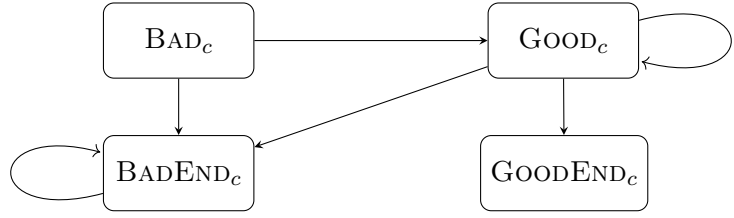

For $\sigma^{(T^*)}_*,\tau^{(T^*)}_*$ to be in stage \Goodend{c}, $\sigma^{(t)}_*,\tau^{(t)}_*$ cannot be in stage \Bad{c} for any $t\ge T$. In other words, because $\sigma^{(T-1)}_*,\tau^{(T-1)}_*$ is in state \Bad{c}, the pair of colorings must escape from \Bad{c} in the very first step of \eqref{extraschedule} and never return. We first show this probability of escape is comparable to the total probability mass of $\mathcal{W}$.


\begin{lem}
Let $\sigma,\tau$ be any neighboring coloring pair in state \Bad{c}, and let $\sigma',\tau'$ be derived from one step of greedy coupling. Then $P[\sigma',\tau' \ \text{in state \Good{c}}]\ge\frac{4(k - d - 1)}{nk}.$\label{lem:badtogood}
\end{lem}

\begin{proof}
Without loss of generality suppose that $(A_c,B_c,\vec{a}^c,\vec{b}^c) = (7,3,(3,3),(1,1))$. Let $u_1,u_2$ be the two $c$-colored neighbors of $v$, and denote the elements of $S_{\tau}(u_1,\sigma(v))$ and $S_{\tau}(u_2,\sigma(v))$ by $\{u_1,w^1_1,w^2_1\}$ and $\{u_2,w^1_2,w^2_2\}$ respectively. We know that the vertices $\{w^1_1,w^2_1,w^1_2,w^2_2\}$ are all distinct. With probability $\frac{4}{n}\cdot\frac{k - d - 1}{k}$, the pair of flips $(S,S')$ chosen under the greedy coupling satisfies $S = S' = S_{\sigma}(w^i_j,c')$ for some $i,j\in\{1,2\}$ and $c'\in A_{\sigma}(w^i_j)\backslash\{\sigma(w^i_j)\}$ (note that $A_{\sigma}(w^i_j)\backslash\{\sigma(w^i_j)\}$ contains neither $\sigma(v)$ nor c). In this case, the flips are just of vertex $w^i_j$ from color $\sigma(w^i_j)$ to a different color not already present in its neighborhood, so the neighboring coloring pair $\sigma',\tau'$ resulting from the flips is in state \Good{c}.
\end{proof}

Once a pair of colorings has escaped from state \Bad{c} into state \Good{c}, at every step it can only stay at \Good{c}, end at stage \Goodend{c}, or get absorbed into stage \Badend{c}. We show that the last two events have probability $\Omega(1/n)$ and $O(1/n)$ respectively.

\begin{lem}
Let $\sigma,\tau$ be any neighboring coloring pair in stage \Good{c}, and let $\sigma',\tau'$ be derived from one step of greedy coupling. Then $P[\sigma',\tau' \ \text{in stage \Goodend{c}}]\ge\frac{k - d - 2}{nk}.$\label{lem:goodtogoodend}
\end{lem}

\begin{proof}
By Lemma~\ref{lem:probend}, the probability that the next pair of flips $(S,S')$ chosen under the greedy coupling is terminal is at least $\frac{k - d - 2}{nk}$.
\end{proof}

\begin{lem}
Let $\sigma,\tau$ be any neighboring coloring pair in stage \Good{c}, and let $\sigma',\tau'$ be derived from one step of greedy coupling. Then $P[\sigma',\tau' \ \text{in stage \Badend{c}}]\le\frac{5}{n}.$\label{lem:goodtobadend}
\end{lem}

\begin{proof}
First note that in order for $\sigma',\tau'$ to be in stage \Badend{c} given that $\sigma,\tau$ were in stage \Good{c}, it must be that condition (iii) of Definition~\ref{def:stages} holds. Furthermore, the pair of components $(S,S')$ flipped to get from $\sigma,\tau$ to $\sigma',\tau'$ cannot be terminal, so $S = S'$.

Suppose that $\delta_c > 2$. In this case, the probability that $\sigma,\tau$ leave stage \Good{c} for stage \Badend{c} is at most the probability that enough $c$-colored neighbors of $v$ are flipped so that $\delta_c$ becomes at most 2. An alternating component $S$ outside of $D_c$ and containing at least $(\delta_c - 2)$ $c$-colored neighbors of $v$ must be flipped in both colorings to achieve this, and the probability the greedy coupling chooses any particular such $(S,S)$ is $p_{\delta_c - 2}/(nk)$. The number of such alternating components is at most $\delta_c\cdot(k - 2)$, so by a union bound the probability that $\delta_c$ becomes 2 is at most $\frac{\delta_c p_{\delta_c - 2}\cdot (k-2)}{nk} < 3/n$.

On the other hand, if $\delta_c < 2$, then $\sigma,\tau$ are not in stage \Good{c} to begin with. So for the rest of the proof, we consider the case of $\delta_c = 2$. We will proceed by casework on $(A_c,B_c,\vec{a}^c,\vec{b}^c)$, which we will denote as $(A,B,(a_1,a_2),(b_1,b_2))$ for simplicity.

Let $\mathcal{E}$ denote the event that $\sigma,\tau$ transition to stage \Badend{c}. Denote the two $c$-colored neighbors of $v$ by $u_1, u_2$. We have that $\mathcal{E}\subseteq \mathcal{E}_1\cup\mathcal{E}_2$, where $\mathcal{E}_1$ is the event that $u_1$ or $u_2$ is flipped in both colorings to a new color, and $\mathcal{E}_2$ is the event that $u_1$ or $u_2$ are not flipped but $\sigma,\tau$ nevertheless transition to stage \Badend{c}. Obviously $\Pr[\mathcal{E}_1]\le 2/n$. We now proceed to bound $\Pr[\mathcal{E}_2]$.

\begin{case}
	If $a_i > 3$ or $b_i > 3$ for some $i = 1,2$, then $\Pr[\mathcal{E}_2]\le\frac{3}{n}$.\label{case:worst}
\end{case}


\begin{proof}
	Without loss of generality, say that $a_1 > 3$. From the vertices of $S_{\tau}(u_1,\sigma(v))$ pick out $w,w'\neq u_1$ such that $w,w',u_1$ form an alternating component. We have that event $\mathcal{E}_2\subseteq\mathcal{A}\cup\mathcal{B}$, where $\mathcal{A}$ is the event that all vertices in $S_{\tau}(u_1,\sigma(v))\backslash\{u_1,w,w'\}$ are flipped so that $S_{\tau'}(u_1,\sigma(v))\subseteq\{u_1,w,w'\}$, and $\mathcal{B}$ is the event that $w$ or $w'$ is flipped and no longer belongs to $S_{\tau'}(u_1,\sigma(v))$. Obviously $\Pr[\mathcal{B}]\le 2/n$. For $\mathcal{A}$, the $(a_1 - 3)$ neighbors of $u_1,w,w'$ in $S_{\tau}(u_1,\sigma(v))$ must be flipped at once, which by a union bound occurs with probability at most $\frac{1}{n}\cdot(a_1 - 3)\cdot p_{a_1 - 3}\le\frac{1}{n}$, where the inequality follows by \eqref{eq:pap}. So $\Pr[\mathcal{E}_2]\le\Pr[\mathcal{A}] + \Pr[\mathcal{B}]\le 3/n$.
\end{proof}

\begin{case}
	If $a_i = 0$ for some $i$ and $b_1,b_2\le 3$, or if $b_i = 0$ for some $i$ and $a_1,a_2\le 3$, then $P[\mathcal{E}_2]\le\frac{1}{n}$.
\end{case}

\begin{proof}
	Suppose without loss of generality that $a_1 = 0$ and $b_1,b_2\le 3$. By the definition of the greedy coupling and the fact that $c\neq\sigma(v),\tau(v)$, $a_1 = 0$ if and only if $S_{\tau}(u_2,\sigma(v))$ consists of $u_1,u_2,w$ for some $\sigma(v)$-colored $w\in N(u_1)\cup N(u_2)$. So $\mathcal{E}_2$ is a subset of the event that $w$ is flipped to any other color. Thus, $\Pr[\mathcal{E}_2]\le\frac{1}{n}$.
\end{proof}

\begin{case}
	If $1\le a_1,a_2,b_1,b_2\le 3$, and if $(a_1,a_2)$ and $(b_1,b_2)$ are both not among $\{(1,1),(3,3)\}$, then $P[\mathcal{E}_2]\le\frac{48}{nk}$.\label{case:insignificant}
\end{case}

\begin{proof}
Suppose $(a_1,a_2)$ and $(b_1,b_2)$ are both not among $\{(1,1),(3,3)\}$. Then $\mathcal{E}_2$ is a subset of the event that the pair of flips $(S,S)$ chosen increases or decreases at least one of $a_1,a_2$ and decreases or increases at least one of $b_1,b_2$, respectively. But for a flip $S$ to decrease some $a_i$ for $i\in\{1,2\}$, it must contain a member of $S_{\tau}(u_i,\sigma(v))$, and for a flip $S$ to increase some $b_j$ for $j\in\{1,2\}$, it must contain the color $c$ or $\tau(v)$. There are at most 3 members of $S_{\tau}(u_i,\sigma(v))$, so the probability of $(S,S)$ both increasing $a_i$ and decreasing $b_j$ is at most $\frac{3}{n}\cdot\frac{2}{k} = \frac{6}{nk}$, and by a union bound over the eight different choices of $i,j$, and increasing/decreasing, we conclude that $\Pr[\mathcal{E}_2]\le\frac{48}{nk}$.
\end{proof}

\begin{case}
	If $1\le a_1,a_2,b_1,b_2\le 3$ and exactly one of the tuples $(a_1,a_2)$ and $(b_1,b_2)$ is among $\{(1,1),(3,3)\}$, then $P[\mathcal{E}_2]\le\frac{4d + 48}{nk}$.
\end{case}

\begin{proof}
Suppose $(b_1,b_2) = (1,1)$. $\mathcal{E}_2\subseteq\mathcal{X}\cup\mathcal{Y}$, where $\mathcal{X}$ is the event that the pair of flips $(S,S)$ chosen increases or decreases some $a_i$ and decreases or increases some $b_i$, respectively, and $\mathcal{Y}$ is the event that $(a_1,a_2)$ becomes $(3,3)$. We already know by Case~\ref{case:insignificant} that $\Pr[\mathcal{X}]\le \frac{48}{nk}$. Supposing without loss of generality that $a_1 < 3$, the event $\mathcal{Y}$ is a subset of the event that a neighbor of a vertex in $S_{\tau}(u_1,\sigma(v))$ is flipped to the color $c$ or $\sigma(v)$. There are at most $2d$ such neighbors, so $\Pr[\mathcal{Y}]\le\frac{2d}{n}\cdot\frac{2}{k} = \frac{4d}{nk}$, and thus $\Pr[\mathcal{E}_2]\le\frac{4d + 48}{nk}$.

Now suppose $(b_1,b_2) = (3,3)$. $\mathcal{E}_2\subseteq\mathcal{X}\cup\mathcal{Z}$ where $\mathcal{X}$ is the event defined above and $\mathcal{Z}$ is the event that $(a_1,a_2)$ becomes $(1,1)$. Supposing without loss of generality that $a_1 > 1$, $\mathcal{Z}$ is a subset of the event that one of the members of $S_{\tau}(u_1,\sigma(v))$ other than $u_1$ is flipped. There are at most two such vertices, so $\Pr[\mathcal{Z}]\le 2/n$ and $\Pr[\mathcal{E}_2]\le\frac{2k + 48}{nk}$.
\end{proof}

\begin{case}
	If $1\le a_1,a_2,b_1,b_2\le 3$ and $(a_1,a_2,b_1,b_2) = (1,1,1,1)$, then $P[\mathcal{E}_2]\le\frac{4d}{nk}$.
\end{case}

\begin{proof}
	$\mathcal{E}_2$ is a subset of the event that one of the neighbors of $u_1$ or $u_2$ is flipped to the color $\sigma(v)$ or $\tau(v)$, so $\Pr[\mathcal{E}_2]\le\frac{2d}{n}\cdot\frac{2}{k} =\frac{4d}{nk}$.
\end{proof}

\begin{case}
	If $1\le a_1,a_2,b_1,b_2\le 3$ and $(a_1,a_2,b_1,b_2) = (3,3,3,3)$, then $P[\mathcal{E}_2]\le\frac{2}{n}$.
\end{case}

\begin{proof}
	$\mathcal{E}_2\subseteq\mathcal{S}\cup\mathcal{T}$, where $\mathcal{S}$ (resp. $\mathcal{T}$) is the event that all $\sigma(v)$-colored (resp. $\tau(v)$-colored) neighbors in $N(u_1)\cup N(u_2)$ in $\tau$ (resp. $\sigma$) are flipped to a different color. Consider an arbitrary $\sigma$-colored neighbor $w$ of $u_1$. $\mathcal{S}$ is a subset of the event that $w$ is flipped, so $\Pr[\mathcal{S}]\le 1/n$. We can bound $\Pr[\mathcal{T}]$ similarly, so $\Pr[\mathcal{E}_2]\le 2/n$.
\end{proof}

Of the upper bounds on $\Pr[\mathcal{E}_2]$ in all of the above cases, the bound of $3/n$ from Case~\ref{case:worst} is the greatest when $k\ge 1.833d$, completing the proof of Lemma~\ref{lem:goodtobadend}.\end{proof}

We are now ready to complete the proof of the main result of this section. 

\begin{proof}[Proof of Lemma~\ref{lem:main}]
Starting from $\sigma^{(T-1)},\tau^{(T-1)}$ in stage \Bad{c}, by Lemma~\ref{lem:badtogood}, the probability of transitioning to stage \Good{c} in the very next step is at least $\frac{4(k - d - 1)}{nk}$. As shown in Figure~\ref{fig:transitions}, once we leave stage \Bad{c} we never return. From stage \Good{c}, it is at most $\frac{5}{n}\cdot\frac{nk}{k - d - 2} = \frac{5k}{k - d - 2}$ times as likely to eventually end up at stage \Badend{c} as it is to end up at stage \Goodend{c}, by Lemmas~\ref{lem:goodtogoodend} and Lemma~\ref{lem:goodtobadend}. So the probability of ending in stage \Goodend{c} is at most $\frac{k - d - 2}{5k + (k - d - 2)}\cdot\frac{4(k - d - 1)}{nk}$, and we conclude that \eqref{eq:mainmarkov} and consequently Lemma~\ref{lem:main} hold for $$\gamma = \frac{(6k - d - 2)(k + 2p_2d)}{4(k - d - 2)(k - d - 1)}$$ as claimed.
\end{proof}

We now finish the proof of Theorem~\ref{thm:main}.

\begin{proof}[Proof of Theorem~\ref{thm:main}]
Note that for $k > 1.833d$ and $p_2 < 0.3$, $\gamma = \frac{(6k - d - 2)(k + 2p_2d)}{4(k - d - 2)(k - d - 1)}< 7.683410$, while $C = \frac{k + 2p_2d}{k - d - 2} < 2.920764$, so $C\gamma < 25.597784$ as defined in Lemma~\ref{lem:reduction}. Thus, substituting $25.597784$ into the $\gamma$ parameter for Linear Program~\ref{def:mixedlp} and solving numerically\footnote{Code for solving Linear Program~\ref{def:mixedlp} can be found at \url{https://github.com/sitanc/mixedlp}.}, we find that for $$p_1 = 1, p_2\approx 0.296706, p_3\approx 0.166762, p_4\approx 0.101790, p_5\approx 0.058475, p_6\approx 0.025989, p_{\alpha} = 0 \ \forall \alpha\ge 7,$$ Linear Program~\ref{def:mixedlp} attains value $\lambda^* < 1.833239$. So provided $k\ge1.833239 d$, Lemma~\ref{lem:reduction} implies that $$1 - \E[d(\sigma^{(\Tstop)},\tau^{(\Tstop)})]\le\frac{k-\lambda^*\cdot d}{k - d - 2} := \alpha$$ for some absolute constant $\alpha>0$. 

For $k\ge 1.833239 d$, Corollary~\ref{eq:maxET} implies that $\beta$ in the definition of Theorem~\ref{thm:hayesvigoda} is at most $\frac{nk}{k - d - 2}\le 2.21n$, so applying Theorem~\ref{thm:hayesvigoda} with $\beta = 2.21n$, and $W = 2N_{max} + 1 = 13$ gives that $$\tau_{mix}\le 2\left\lceil 26(2.21n)/\alpha\right\rceil\left\lceil\ln(n)/\alpha\right\rceil = O(n\log n)$$ as claimed.
\end{proof}

Theorem~\ref{thm:compareglauber} follows as a corollary of Theorem~\ref{thm:main}.

\begin{proof}[Proof of Theorem~\ref{thm:compareglauber}]
	In \cite{vigoda2000improved}, Vigoda proves using the comparison theorem of Diaconis and Saloff-Coste \cite{diaconis1993comparison} that if the flip dynamics mix in time $\tau$, then the Glauber dynamics for sampling colorings mix in time $O(\tau\cdot\log|\Omega|) = O(\tau\cdot k\log n)$.
\end{proof}



\bibliographystyle{alpha}
\bibliography{biblio}

\appendix


\section{Proof of Observation~\ref{obs:slack}}
\label{app:obsproof}

\begin{proof}
	The tightness of \eqref{eq:pap} only for $\alpha = 1$ is obvious. That the other constraints mentioned in the observation have zero slack can be checked by hand. We verify that all other constraints have nonzero slack.

	\setcounter{case}{0}
	\begin{case}
		Constraint \eqref{eq:mainconstraint} for $m = 1$
	\end{case}

	We first consider realizable $(A,B,\vec{a}^c,\vec{b}^c)$. It is easy to see that $(i - 1)(p_i - p_{i+1})\le 71/500$ with equality if and only if $2\le i\le 4$, and that $i(p_i - p_{i+1})\le 1037/1500$ with equality if and only if $i = 1$. Note that for $m = 1$, \begin{align*}H(A,B,\vec{a},\vec{b}) &= \max\left(a_1(p_{a_1} - p_{a_1+1}) + (b_1 - 1)(p_{b_1} - p_{b_1 + 1}), (a_1 - 1)(p_{a_1} - p_{a_1+1}) + b_1(p_{b_1} - p_{b_1 + 1})\right) \\ &\le\frac{71}{500} + \frac{1037}{1500} = \frac{5}{6},\end{align*} with equality if and only if $a_1 = 1$ and $2\le b_1\le 4$ or $b_1 = 1$ and $2\le a_1\le 4$.

	\begin{case}
		Constraint \eqref{eq:mainconstraint} for $m = 2$
	\end{case}

	We analyze this case in the same way that Claim 6 of \cite{vigoda2000improved} is proved. Assume without loss of generality that $p_{a_{max}} - p_A \le p_{b_{max}} - p_B$ and $a_1\ge a_2$. In \cite{vigoda2000improved} it is noted that one may assume that $b_2\ge b_1$ so that $$H(A,B,\vec{a},\vec{b}) = (A - 2a_1)p_A + (B-2b_2 - 1) + (a_1 - 1)p_{a_1} + a_2p_{a_2} + b_1p_{b_1} + b_2p_{b_2} - \min(p_{a_2},p_{b_2} - p_B).$$ Now we proceed by casework on $\min(p_{a_2},p_{b_2} - p_B)$:

	\begin{itemize}
	 	\item $p_{a_2}\le p_{b_2} - p_B$: In this case we have \begin{equation}H(A,B,\vec{a},\vec{b}) = (a_1 - 1)p_{a_1} + (a_2 - 1)p_{a_2} + (A-2a_1)p_A + b_1p_{b_1} + b_2p_{b_2} + (B - 2b_2 - 1)p_B.\label{eq:case1}\end{equation} One can check that $(a - 1)p_{a}\le 1/3$ with equality if and only if $a = 3$. Furthermore, when $a_1 = 3$, $(A-2a_1)p_A \le 0$ with equality if and only if $6\le A\le 7$; when $a_1\neq 3$, $(A-2a_1)p_A > 0$ if and only if $a_1 = a_2$ and $A = 2a_1 + 1$. But for the latter case, $(a_1 - 1)p_{a_1} + (a_2 - 1)p_{a_2} + (A-2a_1)p_A < 2/3$, so we conclude that for any fixed $b_1,b_2,B$, \eqref{eq:case1} is only maximized when $a_1 = a_2 = 3$ and $6\le A\le 7$. In a similar manner, we can verify that for any fixed $a_1,a_2,A$, \eqref{eq:case1} is only maximized when $b_1 = b_2 = 1$ and $B = 3$.

	 	\item $p_{a_2}> p_{b_2} - p_B$: In this case we have \begin{equation}
	 		H(A,B,\vec{a},\vec{b}) = (a_1 - 1)p_{a_1} + a_2p_{a_2} + (A - 2a_1)p_A + b_1p_{b_1} + (b_2 - 1)p_{b_2} + (B-2b_2)p_B.\label{eq:case2}\end{equation} This is symmetric with respect to flipping the roles of $(a_1,a_2)$ and $(b_2,b_1)$, and it can be verified that $(a_1 - 1)p_{a_1} + a_2p_{a_2} + (A-2a_1)p_A\le 4/3$. On the other hand, as we have seen in above, for $(a_1, a_2,b_1,b_2,A,B) = (7,3)$, we have that $H(A,B,\vec{a},\vec{b}) = 8/3$, concluding the proof for $m = 2$.
	\end{itemize}

	\begin{case}
		Constraint \eqref{eq:Hforsigmav}
	\end{case}

	For $m = 2$, the left-hand side of \eqref{eq:Hforsigmav} is $b_1p_{b_1 + b_2} + b_1p_{b_1} + b_2p_{b_2}$, which attains its maximum value of $p_2 + 2 < -1 + 2\lambda$ at $b_1 = b_2 = 1$. For $m > 2$, note that \eqref{eq:pap} implies that the left-hand side of \eqref{eq:Hforsigmav} is at most $m + 1 < -1 + \lambda\cdot m$ provided $\lambda > 5/3$, which is certainly the case.

	\begin{case}
	 	Constraint \eqref{eq:approxconstraint}
	\end{case}

	One can check that $(A-2)p_A\le 287/1500$. And if $p_a\le p_b$, then $a\cdot p_a + b\cdot p_b - \min(p_a,p_b) = (a-1)p_a + b\cdot p_b$. But $(a-1)p_a\le 1/3$ and $b\cdot p_b\le 1$. So $x^* = 287/1500$ and $y^* = 4/3$, and it is clear that $-1 + (11/6)\cdot 3 > 2\cdot x^* + m^*\cdot y^*$ for $m^* = 3$, so \eqref{eq:approxconstraint} has nonzero slack.
\end{proof}

\end{document}